\documentclass[]{interact}

\usepackage{epstopdf}
\usepackage[caption=false]{subfig}

\usepackage[natbibapa,nodoi]{apacite}
\theoremstyle{plain}
\newtheorem{theorem}{Theorem}[section]
\newtheorem{lemma}[theorem]{Lemma}

\theoremstyle{definition}

\theoremstyle{remark}
\newtheorem{remark}{Remark}

\begin{document}

\articletype{ARTICLE TEMPLATE}

\title{Handling actuator magnitude and rate saturation in uncertain over-actuated systems: A modified projection algorithm approach}

\author{
\name{Seyed Shahabaldin Tohidi\textsuperscript{a}\thanks{CONTACT S.~S. Tohidi. Email: shahabaldin@bilkent.edu.tr} and Yildiray Yildiz\textsuperscript{a}\thanks{CONTACT Y. Yildiz. Email: yyildiz@bilkent.edu.tr}}
\affil{\textsuperscript{a}Faculty of Mechanical Engineering, Bilkent University, Cankaya, Ankara 06800, Turkey}
}

\maketitle

\begin{abstract}
This paper proposes a projection algorithm which can be employed to bound actuator signals, in terms of both magnitude and rate, for uncertain systems with redundant actuators. The investigated closed loop control system is assumed to contain an adaptive control allocator to distribute the total control input among actuators. Although conventional control allocation methods can handle actuator rate and magnitude constraints, they cannot consider actuator uncertainty. On the other hand, adaptive allocators manage uncertainty and actuator magnitude limits. The proposed projection algorithm enables adaptive control allocators to handle both magnitude and rate saturation constraints. A mathematically rigorous analysis is provided to show that with the help of the proposed projection algorithm, the performance of the adaptive control allocator can be guaranteed, in terms of error bounds. Simulation results are presented, where the Aero-Data Model In Research Environment (ADMIRE) is used as an over-actuated system, to demonstrate the effectiveness of the proposed method.
\end{abstract}


\section{Introduction}

Actuator constraints such as magnitude and rate limits play a prominent role in advanced control systems.
These limits induce nonlinear behavior which may lead to performance degradation, occurrence of limit cycles, multiple equilibria, and even instability \citep{Khal02,TarGar11}. Actuator rate limits, specifically, introduce phase lags, which act as time delays, that can lead to persistent undesired oscillations called Pilot Induced Oscillations (PIO) \citep{TohYil18,QueTar17,YilKol10,Yil11a,Yil11b,YilKolAco11,AcoYil14}. These oscillations generally occur due to an abnormal coupling between the pilot and the aircraft, instigated by various factors such as high pilot gains, actuator rate saturation and control mode switches \citep{Mcr95}.

For systems with uncertainties, various adaptive controllers that account for actuator magnitude limits exist in the literature \citep{Gru19, KarAnnas93, LavHov07a, LavHov07b}. There are also adaptive approaches related to the problem of handling actuators that are constrained in both magnitude and rate. In the paper by \cite{YonFra14}, the approach presented by \cite{LavHov07a} and \cite{LavHov07b} is extended for systems with rate and magnitude limits. In the method proposed by \cite{LeonHad09}, the reference inputs as well as the control signals are modified adaptively in order to guarantee the stability in the presence of magnitude and rate limits. In a recent work by \cite{GauAnnas19}, plant dynamics is augmented with the actuator dynamics, and an adaptive controller is introduced to compensate the effect of actuator magnitude and rate limits.

With the reduction of actuator costs due to advances in microprocessors, and with the help of actuator miniaturization, the utilization of redundant actuators have been growing in recent years. Actuator redundancy can improve the performance, maneuverability and the ability to tolerate system faults. The process of distributing control signals among redundant actuators is performed by control allocation. A study on control allocation that considers actuator magnitude constraints is conducted by \cite{Dur93} by using direct allocation method. Daisy chain control allocation method, which handles actuator magnitude limit, is employed by \cite{Buf97}. Actuator magnitude saturation of an unmanned underwater vehicle is considered using pseudo inverse based control allocation (\cite{Omer07}). An iterative approach based on the null space of the control matrix is proposed by \cite{Toh16}, which handles actuator magnitude limits. Optimization based control allocation is one of the most common methods of accounting for actuator magnitude and rate constraints \citep{PetBod06, Har05, Har02, Joh08, Yil11a, Yil11b, Safa19}. A sequential algorithm to solve optimization based control allocation is proposed by \cite{Nas17}. A survey on control allocation methods can be found in the study conducted by \cite{JohFos13}. A recent control allocation study is presented by \cite{NadSed19}, where model predictive control is employed to handle actuator magnitude constraints.


When a system has uncertain dynamics, together with redundant actuators, it is natural to consider an adaptive control allocator to achieve the task of distributing the total control effort among actuators. There exits few approaches presented in the literature that addresses the topic of adaptive control allocation. The method proposed by \cite{TjoJoh08} reduces the difference between virtual and actual control signals, and guarantees that the control signals ultimately converge to an optimal set. An adaptive control allocation for a hexacopter system is proposed by \cite{FalHol16}. A model reference adaptive control allocation structure is proposed by \cite{TohYil16}. This method is also extended to handle actuator magnitude limits \citep{TohYil17, TohYil19, TohYil20}.  

Projection algorithm is an appealing approach in robust adaptive control design. Restricting adaptive parameters while ensuring the stability of the closed loop system, simultaneously, is a prominent benefit of employing this algorithm in adaptive systems. It is noted that existing projection algorithms (\cite{Pra91, Lav13}) bound adaptive parameters' magnitudes and thus do not have a straightforward utility to handle actuator rate limits.
In this paper, we propose a projection algorithm that can be used in adaptive control allocation implementations, where actuators are both magnitude and rate limited. Therefore, the contribution of this paper is a projection algorithm that can handle magnitude and rate limited redundant actuators for systems with uncertain dynamics, where a control allocator is utilized in the controller structure. We show that, the existence and uniqueness of the solution of the differential equation describing the proposed projection algorithm can be guaranteed. Furthermore, we provide a performance guarantee, in terms of error bounds, for the exploited adaptive control allocation, which is possible thanks to the proposed projection algorithm. 

To summarize, we propose an answer to this question: ``How can we modify the conventional projection algorithm, so that we can employ it in adaptive control allocation implementations where actuators are both magnitude and rate saturated?" To the best of our knowledge, this question is not answered earlier. It needs to be emphasized that a control allocator is not a controller and cannot be replaced as a controller. The duty of the control allocation is distributing the controller signal, or the total control input, among redundant actuators. The method proposed in this paper is for the systems where an adaptive control allocator is used in the loop. We are not proposing a new controller or a new control allocation method. 


This paper is organized as follow. Notations used throughout the paper and the conventional, element-wise projection algorithm and its properties are given in Section \ref{notation}. Section \ref{sec4} presents the uncertain over-actuated system along with the adaptive control allocation utilizing the conventional projection algorithm. The proposed modified projection algorithm and its characteristics are presented in Section \ref{sec3}. The ADMIRE model is used in Section \ref{sec6} to illustrate the effectiveness of the proposed methodology in the simulation environment. Finally, a summary is given in Section
\ref{sec7}.


\section{Notations and preliminaries}\label{notation}

Throughout this work, $ \mathbb{R} $ is the set of real numbers, $ \mathbb{R}^+ $ is the set of positive real numbers, $ \mathbb{R}^m $ is a column vector with $ m $ real elements and $ \mathbb{R}^{m\times n} $ is an $ m\times n $ matrix of real elements.  $ ||.|| $ refers to the Euclidean norm for vectors and induced 2-norm for matrices, and $ ||.||_F $ refers to the Frobenius norm. $ I_r $ is the identity matrix of dimension $ r\times r $, $ 0_{r\times n} $ is the zero matrix of dimension $ r\times n $, and $\text{tr}(.)$ refers to the trace operation. The over-dot notation will be used for time derivatives only, i.e. $ \dot{(\cdot)} = d(\cdot)/dt $.

Consider $ Y\in \mathbb{R}^{r \times m}$ and ${\theta}_v\in \mathbb{R}^{r\times m}$. 
The element-wise projection operator $ \text{Proj}(.,.):\mathbb{R}\times \mathbb{R}\rightarrow \mathbb{R} $ is defined as
\begin{align} \label{eq:17n}
\text{Proj}(\theta_{v_{i,j}}, Y_{i,j})\equiv \left \{\begin{array}{l}Y_{i,j}-Y_{i,j}f_{i,j}\ \ \ \ \text{if}\ f_{i,j}>0 \ \ \ \& \ \ Y_{i,j}( \frac{df_{i,j}}{d\theta_{v_{i,j}}}) >0\\  Y_{i,j} \ \ \ \ \ \ \ \ \ \ \ \ \ \ \ \ \text{otherwise},\end{array}\right.
\end{align}
where $\theta_{v_{i,j}}$ and $Y_{i,j}$ refer to the element in the  $i^{\text{th}}$ row and $ j^{\text{th}} $ column of $\theta_{v}$ and $Y$, respectively, and where $f_{i,j}(.):\mathbb{R}\rightarrow \mathbb{R}$ is a convex and continuously differentiable function defined as
\vspace{-0.25cm}
\begin{equation} \label{eq:18n}
f_{i,j}=f(\theta_{i,j})=\frac{(\theta_{v_{i,j}}-\theta_{min_{i,j}}-\zeta_{i,j})(\theta_{v_{i,j}}-\theta_{max_{i,j}}+\zeta_{i,j})}{(\theta_{max_{i,j}}-\theta_{min_{i,j}}-\zeta_{i,j})\zeta_{i,j}},
\end{equation}
where $ \zeta_{i,j}\in \mathbb{R}^+ $ is the projection tolerance of $\theta_{v_{i,j}}$ such that $ \zeta_{i,j}<0.5(\theta_{max_{i,j}}-\theta_{min_{i,j}}) $, $\theta_{max_{i,j}}-\zeta_{i,j}>0$ and $\theta_{min_{i,j}}+\zeta_{i,j}<0$. $\theta_{max_{i,j}}>0$ and $\theta_{min_{i,j}}<0$ are the upper and lower bounds of the $ (i,j)^{\text{th}} $ element of $\theta_{v}$. Therefore, the projection operator $ \text{Proj}(\theta_{v}, Y) $ operates on the elements of $ \theta_{v} $ and $ Y $ using (\ref{eq:17n}) and (\ref{eq:18n}).

 The following lemmas are useful in proving the main theorems where projection algorithm is used (\cite{Lav13, Pra91, NarAnn12}).
\begin{lemma}\label{lem1}
	If an adaptive algorithm with adaptive law $\dot{\theta}_{v_{i,j}}=\text{Proj}(\theta_{v_{i,j}},Y_{i,j})$ and initial conditions $\theta_{v_{i,j}}(0)\in \Omega_{i,j}=\{\theta_{v_{i,j}}\in \mathbb{R}|f(\theta_{v_{i,j}})\leq 1\}$, where $f(\theta_{v_{i,j}}):\mathbb{R}\rightarrow \mathbb{R}$ is defined as in (\ref{eq:18n}), then $\theta_{v_{i,j}}\in \Omega_{i,j}$ for $\forall t\geq 0$.
\end{lemma}
\begin{proof}
	The proof of Lemma \ref{lem1} can be found in \cite{Lav13}.
\end{proof}
\begin{lemma} \label{lem2}
	Let $ \theta_{v_{i,j}}^*\in [\theta_{min_{i,j}}+\zeta_{i,j},\ \theta_{max_{i,j}}-\zeta_{i,j}] $, and consider the projection algorithm in (\ref{eq:17n}) with convex function (\ref{eq:18n}), the following inequality holds:
	\begin{equation}\label{eq:e20xx}
	\begin{array}{ll}
	tr( ({\theta}_v^T-{{\theta}_v^*}^T) ( -Y + \text{Proj}(\theta_v, Y) ) ) \leq 0.
	\end{array}
	\end{equation}
\end{lemma}
\begin{proof}
		The proof of Lemma \ref{lem2} can be found in \cite{Lav13}.
\end{proof}


\section{Problem statement}\label{sec4}
In this section, firstly, the over-actuated plant with constrained uncertain actuators is introduced. Then, the adaptive control allocation utilizing the conventional projection algorithm (\ref{eq:17n}), which can bound only the magnitude of actuators input signals, is presented. Finally, the problem statement motivating the proposed projection algorithm is given.

Consider the following uncertain over-actuated plant dynamics
\begin{align}\label{eq:e1x}
\ \ \dot{x}&=Ax+B_u\Lambda u\notag \\
&=Ax+B_vB\Lambda u\notag \\
&=Ax+B_vv_s,
\end{align}
where $x\in \mathbb{R}^{n}$ is the state vector, $u=[u_1,...,u_m]^T\in \mathbb{R}^{m}$ is the magnitude constrained actuator command vector, where $u_j\in [u_{\text{min}_{j}},u_{\text{max}_{j}}]$ with $ u_{\text{max}_{j}}>0 $ and $ u_{\text{min}_{j}}<0 $. The matrix $A \in \mathbb{R}^{n\times n}$ is the known state matrix and $B_u=B_vB \in \mathbb{R}^{n\times m}$ is the known rank deficient control input matrix which is decomposed into the known matrices $B_v \in \mathbb{R}^{n\times r}$ and $B \in \mathbb{R}^{r \times m}$ such that $ rank(B)=rank(B_v)=r $. The actuator loss of effectiveness is modeled as a diagonal matrix $\Lambda \in \mathbb{R}^{m\times m}$ with uncertain positive elements. The goal of the static control allocation methods in the absence of uncertainty, where $ \Lambda=I_m $, is to distribute the total control effort $ v_s\in \mathbb{R}^r $, produced by a controller, to the redundant actuators such that $ Bu=v_s $. 
In the presence of uncertainty, the static control allocation methods are not applicable since the goal of the control allocation becomes
\begin{equation}\label{eq:e2x}
B\Lambda u=v_s.
\end{equation}
One way to achieve (\ref{eq:e2x}) is by employing the following control allocation system proposed by \cite{TohYil16}
\begin{subequations}\label{eq:aloc}
	\begin{align}
	\dot{\xi}&=A_m\xi+B\Lambda u-v_s,\label{eq:4}\\
	\dot{\xi}_m&=A_m\xi_m,\label{eq:5}\\
	\dot{\theta}_{v}&=g(\theta_{v}, Y(v_s, e)),\label{eq:55}\\
	u&={\theta}_v^Tv_s,\label{eq:6}
	\end{align}
\end{subequations}
where $ \xi\in \mathbb{R}^{r} $ is the output of the virtual dynamics, $ \theta_v\in \mathbb{R}^{r\times m} $ is the adaptive parameter to be updated, $ \xi_m\in \mathbb{R}^{r} $ is the output of the reference model, $ e= \xi-\xi_m $, (\ref{eq:5}) is the reference model with a Hurwitz matrix $ A_m\in \mathbb{R}^{r\times r} $, (\ref{eq:55}) is the adaptive law where $ g(.,.):\mathbb{R}^{r\times m}\times \mathbb{R}^{r\times m}\rightarrow \mathbb{R}^{r\times m} $ is a projection algorithm, and $ u $ is the control allocation signal, or the actuator command signal. It can be shown that (\cite{TohYil16}), in the absence of actuator limits, $ e $ converges to zero and thus the control allocation goal (\ref{eq:e2x}) is achieved. In the presence of actuator magnitude limits, $ e $ converges to a predetermined compact set (\cite{TohYil19, TohYil20}).

In the presence of actuator magnitude limits, if the control signal $ v_s $ is bounded, then (\ref{eq:6}) shows that in order to produce actuator command signals $ u_j, j=1, ..., m $, that respect the actuator saturation bounds, such that $u_j\in [u_{\text{min}_{j}},u_{\text{max}_{j}}]$, the elements of the adaptive parameter matrix $ \theta_v $ should be appropriately bounded. It is shown in \cite{TohYil19, TohYil20} that this could be achieved, together with the stability of the overall system dynamics, by using the conventional projection operator (\ref{eq:17n}) as the function $ g $ in (\ref{eq:55}). 



Problem statement: If the actuators in (\ref{eq:e1x}) are not only magnitude saturated but also rate saturated, i.e. $ \dot{u}_j \in [\dot{u}_{min_{j}}, \dot{u}_{max_{j}}] $, $ j=1,..., m $, how should the projection algorithm (\ref{eq:17n}), which is used as the function $ g $ in (\ref{eq:55}),  be modified to handle this additional condition?

To address the above problem, we need to reconstruct the conventional projection algorithm (\ref{eq:17n}) such that not only the magnitude but also the rate of change of the elements of the matrix $ \theta_v $ become bounded. This problem needs to be solved in such a way that the new projection algorithm must have useful properties similar to the ones given in Lemma \ref{lem1} and Lemma \ref{lem2}, to ensure the stability of the closed loop control system. In the next section, this new projection algorithm is introduced. 


\section{Modified projection algorithm}\label{sec3}

The structure of the overall closed loop control system considered in this paper, consisting of the controller, the control allocator and the plant, is presented in Figure~\ref{fig:figure3}.
\begin{figure}
	\begin{center}
		\vspace{0.2cm}
		\includegraphics[width=12cm]{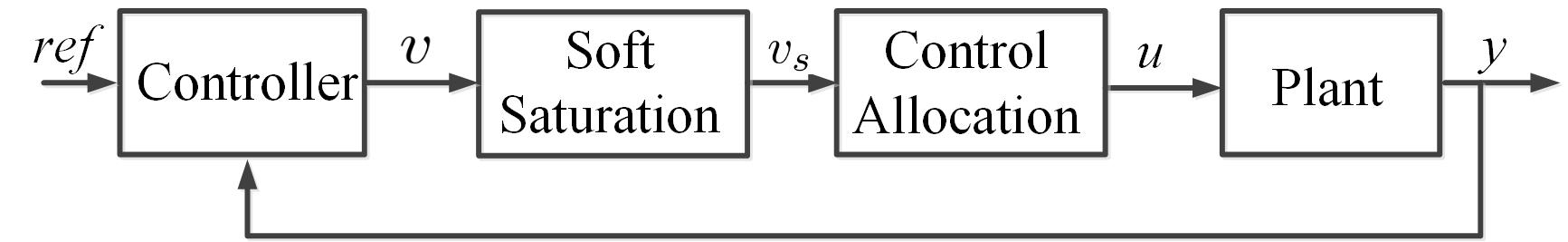}    %
		\caption{Closed loop control system.} 
		\label{fig:figure3}
	\end{center}
\end{figure}
The soft saturation introduced after the controller ensures that the input of the control allocator, $ v_s $, and its derivative, $ \dot{v}_s $, are bounded. From (\ref{eq:55}) and (\ref{eq:6}), it can be seen that one way to obtain a bounded actuator command signal $ u $ is to restrict both the magnitude and the rate of change of the adaptive parameter matrix $ \theta_v $. This restriction must be achieved while ensuring the boundedness of all the signals in the closed loop control system. It is noted that a rate and magnitude bounded total control input $ v_s $ does not guarantee a rate and magnitude bounded actuator input signal vector $ u $, due to the nature of the adaptation in the control allocator. 




The approach proposed in this paper for bounding the adaptive parameter matrix $ \theta_v $ in terms of both magnitude and rate is based on projecting $ Y_{i,j} $ and $ \theta_{v_{i,j}} $, simultaneously. In this method, apart from the function $ f_{i,j} $ introduced in (\ref{eq:18n}), another convex and continuously differentiable function given as
\begin{equation} \label{eq:18nxx}
h_{i,j}=h({Y}_{{i,j}})=\ \frac{({Y}_{{i,j}}-Y_{min_{i,j}}-\epsilon_{i,j})(Y_{{i,j}}-Y_{max_{i,j}}+\epsilon_{i,j})}{(Y_{max_{i,j}}-Y_{min_{i,j}}-\epsilon_{i,j})\epsilon_{i,j}}
\end{equation}
is introduced, where $Y_{max_{i,j}}>0$ and $Y_{min_{i,j}}<0$ are the allowable maximum and minimum bounds of $Y_{{i,j}}$, respectively, and $ \epsilon_{i,j}\in \mathbb{R^+} $ is the projection tolerance such that $ Y_{max_{i,j}}-\epsilon_{i,j}>0 $ and $ Y_{min_{i,j}}+\epsilon_{i,j}<0 $.

Using (\ref{eq:18n}) and (\ref{eq:18nxx}), an element-wise, modified projection algorithm is proposed as
\begin{align} \label{eq:17nx}
\text{Proj}_m(\theta_{v_{i,j}}, Y_{i,j})\equiv \left \{\begin{array}{l}Y_{i,j}(1-\hat{f}_{i,j})(1-\hat{h}_{i,j})\ \ \ \ \text{if}\ f_{i,j}\geq 0 \ \ \ \& \ \ Y_{i,j} \frac{df_{i,j}}{d\theta_{v_{i,j}}}\geq 0\ \ \& \ h_{i,j}\geq 0\\  Y_{i,j}(1-\hat{f}_{i,j}) \ \ \ \ \ \ \ \ \ \ \ \ \ \ \ \ \text{if}\ f_{i,j}>0 \ \ \ \& \ \ Y_{i,j} \frac{df_{i,j}}{d\theta_{v_{i,j}}} >0\\ Y_{i,j}(1-\hat{h}_{i,j}) \ \ \ \ \ \ \ \ \ \ \ \ \ \ \ \  \text{if}\ h_{i,j}>0 \\  Y_{i,j} \ \ \ \ \ \ \ \ \ \ \ \ \ \ \ \ \ \ \ \ \ \ \ \ \ \ \ \ \text{otherwise}, \end{array}\right.
\end{align}
where $ \hat{f}_{i,j}=min\{1,f_{i,j}\} $ and $ \hat{h}_{i,j}=min\{1,h_{i,j}\} $. 

Using this projection algorithm, the adaptive law is given as $ \dot{\theta}_{v_{i,j}}=\text{Proj}_m(\theta_{v_{i,j}},Y_{i,j}) $. In the proposed projection algorithm defined in (\ref{eq:17nx}), when $ \theta_{v_{i,j}} $ reaches its boundary value ($ \theta_{max_{i,j}} $ or $ \theta_{min_{i,j}} $), $ f_{i,j} $ reaches $ 1 $, and from the first and second conditions of (\ref{eq:17nx}), $ \text{Proj}_m(\theta_{v_{i,j}}, Y_{i,j}) $ reaches zero. When $ Y_{{i,j}} $ reaches its boundary value ($ Y_{max_{i,j}} $ or $ Y_{min_{i,j}} $), $ h_{i,j} $ reaches $ 1 $, and from the first and third conditions of (\ref{eq:17nx}), $ \text{Proj}_m(\theta_{v_{i,j}}, Y_{i,j}) $ reaches zero. In addition, since $ f_{i,j} $ and $ h_{i,j} $ cannot exceed one, the magnitude and rate of $ \theta_{v_{i,j}} $ are both bounded. A formal proof is given below, in Lemma \ref{lem5}. It is noted that it is not necessary to take the time derivative of any signal to implement the proposed projection algorithm.



\begin{lemma}\label{lem5}
	Given the adaptive law $\dot{\theta}_{v_{i,j}}=\emph{Proj}_m(\theta_{v_{i,j}}, Y_{i,j})$, where the projection operator is given in (\ref{eq:17nx}), together with convex and continuously differentiable functions (\ref{eq:18n}) and (\ref{eq:18nxx}), if the initial conditions are defined as $\theta_{v_{i,j}}(0)\in \Omega_{i,j}=\{\theta_{v_{i,j}}\in \mathbb{R}|f(\theta_{v_{i,j}})\leq 1\}$ and $Y_{{i,j}}(0)\in \bar{\Omega}_{i,j}=\{Y_{{i,j}}\in \mathbb{R}|
	h(Y_{{i,j}})\leq 1\}$, then $\theta_{v_{i,j}}(t)\in \Omega_{i,j}$ and $Y_{{i,j}}(t)\in \bar{\Omega}_{i,j}$ for all $ t\geq 0$.
\end{lemma}

\begin{proof}
	Taking the time derivative of the convex function $ f(\theta_{v_{i,j}}) $ along the dynamics of $ \theta_{v_{i,j}} $, we have
	\begin{align} 
	\frac{df_{i,j}}{dt}&=\frac{df_{i,j}}{d\theta_{v_{i,j}}}\frac{d\theta_{v_{i,j}}}{dt}=\frac{df_{i,j}}{d\theta_{v_{i,j}}}\text{Proj}_m(\theta_{v_{i,j}}, Y_{i,j})\notag \\ &= \left \{\begin{array}{l}\frac{df_{i,j}}{d\theta_{v_{i,j}}}Y_{i,j}(1-\hat{f}_{i,j})(1-\hat{h}_{i,j})\ \ \ \ \ \text{if}\ f_{i,j}\geq 0 \ \ \ \& \ \ Y_{i,j} \frac{df_{i,j}}{d\theta_{v_{i,j}}}\geq 0 \ \ \& \ h_{i,j}\geq 0 \\  \frac{df_{i,j}}{d\theta_{v_{i,j}}}Y_{i,j}(1-\hat{f}_{i,j}) \ \ \ \ \ \ \ \ \ \ \ \ \ \ \ \ \  \text{if}\ f_{i,j}>0 \ \ \ \& \ \ Y_{i,j} \frac{df_{i,j}}{d\theta_{v_{i,j}}} >0\\ \frac{df_{i,j}}{d\theta_{v_{i,j}}}Y_{i,j}(1-\hat{h}_{i,j}) \ \ \ \ \ \ \ \ \ \ \ \ \ \ \ \ \ \text{if}\ h_{i,j}>0\\  \frac{df_{i,j}}{d\theta_{v_{i,j}}}Y_{i,j} \ \ \ \ \ \ \ \ \ \ \ \ \ \ \ \ \ \ \ \ \ \ \ \ \ \ \ \ \ \text{otherwise} \end{array}\right.\notag
	\end{align}
	\begin{align}
	&\Rightarrow \left \{\begin{array}{l}\frac{df_{i,j}}{dt}=0\ \ \ \ \ \ \ \ \ \ \ \ \ \ \text{if}\ f_{i,j}=1 \ \ \ \ \ \ \ \& \ \ \ \ Y_{i,j} \frac{df_{i,j}}{d\theta_{v_{i,j}}}\geq 0\ \ \ \ \& \ \ \ \ h_{i,j}=1 \\
	\frac{df_{i,j}}{dt}=0\ \ \ \ \ \ \ \ \ \ \ \ \ \ \text{if}\ 0\leq f_{i,j}<1 \ \ \& \ \ \ \ Y_{i,j} \frac{df_{i,j}}{d\theta_{v_{i,j}}}\geq 0\ \ \ \ \& \ \ \ \ h_{i,j}=1 \\
	\frac{df_{i,j}}{dt}=0\ \ \ \ \ \ \ \ \ \ \ \ \ \ \text{if}\ f_{i,j}=1 \ \ \ \ \ \ \ \& \ \ \ \ Y_{i,j} \frac{df_{i,j}}{d\theta_{v_{i,j}}}\geq 0\ \ \ \ \& \ \ \ \ 0\leq h_{i,j}<1\\
	\frac{df_{i,j}}{dt}>0\ \ \ \ \ \ \ \ \ \ \ \ \ \ \text{if}\ 0\leq f_{i,j}<1 \ \ \& \ \ \ Y_{i,j} \frac{df_{i,j}}{d\theta_{v_{i,j}}}\geq 0\ \ \ \ \& \ \ \ \ 0\leq h_{i,j}<1\\
	\frac{df_{i,j}}{dt}=0 \ \ \ \ \ \ \ \ \ \ \ \ \ \ \text{if}\ f_{i,j}=1 \ \ \ \ \ \ \ \ \& \ \ \  Y_{i,j} \frac{df_{i,j}}{d\theta_{v_{i,j}}} >0\\
	\frac{df_{i,j}}{dt}>0 \ \ \ \ \ \ \ \ \ \ \ \ \ \ \text{if}\ 0<f_{i,j}<1 \ \ \ \& \ \ \ Y_{i,j} \frac{df_{i,j}}{d\theta_{v_{i,j}}} >0\\ 
	\frac{df_{i,j}}{dt}=0 \ \ \ \ \ \ \ \ \ \ \ \ \ \ \text{if}\ h_{i,j}=1\\ 
	\end{array}\right.
	\end{align}
	Also, when $ \hat{f}_{i,j}=1 $, $ \frac{df_{i,j}}{dt}=\frac{df_{i,j}}{d\theta_{v_{i,j}}}\text{Proj}_m(\theta_{v_{i,j}}, Y_{i,j})\leq 0 $. Therefore, if $ \theta_{v_{i,j}}(0)\in \Omega_{i,j} $, $\theta_{v_{i,j}}(t)\in \Omega_{i,j}$ for all $ t\geq 0$. The same procedure can be followed for $ \frac{dh_{i,j}}{dt}=\frac{dh_{i,j}}{dY_{i,j}}\frac{dY_{i,j}}{d\theta_{v_{i,j}}}\text{Proj}_m(\theta_{v_{i,j}},Y_{i,j}) $ to prove that if $ Y_{{i,j}}(0)\in \bar{\Omega}_{i,j} $, then $Y_{{i,j}}(t)\in \bar{\Omega}_{i,j}$ for all $ t\geq 0$.
\end{proof}

Below, in Lemma \ref{lem6}, a property of the proposed projection algorithm, which is analogous to Lemma \ref{lem2}, is given, which will be useful later in the stability investigation.

\begin{lemma}\label{lem6}
	Let $ \theta_{v_{i,j}}^*\in [\theta_{min_{i,j}}+\zeta_{i,j}\ \ \theta_{max_{i,j}}-\zeta_{i,j}] $, $Y_{{i,j}}(0)\in \bar{\Omega}_{i,j}=\{Y_{{i,j}}\in \mathbb{R}|
	h(Y_{{i,j}})\leq 1\}$, and consider the projection algorithm (\ref{eq:17nx}) with convex functions (\ref{eq:18n}) and (\ref{eq:18nxx}). The inequality
	\begin{equation}\label{eq:e20}
	\begin{array}{ll}
	tr( ({\theta}_v^T-{{\theta}_v^*}^T) ( -Y + \emph{Proj}_m(\theta_v, Y) ) ) \leq ||\tilde{\theta}_{max}||_F||{Y}_{MAX}||_F
	\end{array}
	\end{equation}
	holds, where $ \tilde{\theta}_{max} $ and $ {Y}_{MAX} $ are the matrices whose elements constitute the upper bounds of the absolute values of the elements of $ \tilde{\theta} $ and $ Y $, respectively.	
\end{lemma}

\begin{proof}
	If $ f_{i,j}\geq 0 $, $ Y_{i,j} (df_{i,j}/d\theta_{v_{i,j}})\geq 0 $ and $ h_{i,j}\geq 0 $ (first condition), then
	\begin{equation}\label{eq:e20pfx}
	\begin{array}{ll}
	&tr\big( ({\theta}_v^T-{{\theta}_v^*}^T) \big( -Y + \text{Proj}_m(\theta_v, Y) \big) \big) \\
	&=\displaystyle \sum_{j=1}^{m} \sum_{i=1}^{r} ({\theta}_{v_{i,j}}-{{\theta}_{v_{i,j}}^*}) \big( -Y_{i,j} + \text{Proj}_m(\theta_{v_{i,j}}, Y_{i,j})\big) \\
	&=\displaystyle \sum_{j=1}^{m} \sum_{i=1}^{r} ({\theta}_{v_{i,j}}-{{\theta}_{v_{i,j}}^*}) \big( -Y_{i,j} + Y_{i,j}(1-\hat{f}_{i,j})(1-\hat{h}_{i,j})\big) \\
	&=\displaystyle \sum_{j=1}^{m} \sum_{i=1}^{r} ({\theta}_{v_{i,j}}-{{\theta}_{v_{i,j}}^*}) \big( -Y_{i,j}\hat{f}_{i,j}-Y_{i,j}\hat{h}_{i,j}+Y_{i,j}\hat{f}_{i,j}\hat{h}_{i,j} )\big).
	\end{array}
	\end{equation}
	$ 0 \leq \hat{h}_{i,j}\leq 1 $ and $ 0\leq \hat{f}_{i,j}\leq 1 $, therefore $ |Y_{i,j}\hat{f}_{i,j}|\geq |Y_{i,j}\hat{f}_{i,j}\hat{h}_{i,j}| $ and $ |Y_{i,j}\hat{h}_{i,j}|\geq |Y_{i,j}\hat{f}_{i,j}\hat{h}_{i,j}| $. Hence, 
	\begin{equation}\label{eq:e20pfxx}
	\begin{array}{ll}
	&\displaystyle \sum_{j=1}^{m} \sum_{i=1}^{r} ({\theta}_{v_{i,j}}-{{\theta}_{v_{i,j}}^*}) \big( -Y_{i,j}\hat{f}_{i,j}-Y_{i,j}\hat{h}_{i,j}+Y_{i,j}\hat{f}_{i,j}\hat{h}_{i,j} )\big)\\
	&\leq \displaystyle \sum_{j=1}^{m} \sum_{i=1}^{r} \underbrace{({{\theta}_{v_{i,j}}^*}-{\theta}_{v_{i,j}}) Y_{i,j}\hat{f}_{i,j}\hat{h}_{i,j}}_{<0}<0.
	\end{array}
	\end{equation}
	If $ h_{i,j}>0 $ (third condition), then
	\begin{equation}\label{eq:e20pfxyz}
	\begin{array}{ll}
	&tr\big( ({\theta}_v^T-{{\theta}_v^*}^T) \big( -Y + \text{Proj}_m(\theta_v, Y) \big) \big) \\
	&=\displaystyle \sum_{j=1}^{m} \sum_{i=1}^{r} ({\theta}_{v_{i,j}}-{{\theta}_{v_{i,j}}^*}) \big( -Y_{i,j} + \text{Proj}_m(\theta_{v_{i,j}}, Y_{i,j})\big) \\
	&=\displaystyle \sum_{j=1}^{m} \sum_{i=1}^{r} ({\theta}_{v_{i,j}}-{{\theta}_{v_{i,j}}^*}) \big( -Y_{i,j} + Y_{i,j}(1-\hat{h}_{i,j})\big) \\
	&\leq \displaystyle \sum_{j=1}^{m} \sum_{i=1}^{r} |{{\theta}_{v_{i,j}}^*}-{\theta}_{v_{i,j}}|Y_{MAX_{i,j}}\\
	&=tr(|\tilde{\theta}_v^T|Y_{MAX})\leq ||\tilde{\theta}_{max}||_F||Y_{MAX}||_F.
	\end{array}
	\end{equation}
	Same procedure used in the proof of Lemma \ref{lem2} can be employed to complete the proof for the second and fourth conditions. 
\end{proof}



Discontinuity in the projection algorithm is not desirable and may cause numerical problems. In the following lemma, we prove that the proposed projection algorithm is continuous.
\begin{lemma}\label{lem7x}
	For continuous $ \theta_{v_{i,j}} $ and $ Y_{i,j} $, the function $\emph{Proj}_m({\theta}_{v_{i,j}}, Y_{i,j}):S_{\theta} \times S_{Y} \rightarrow \mathbb{R}$, where $ S_{\theta}, S_{Y}\subset \mathbb{R} $, is continuous.
\end{lemma}
\begin{proof}
We first decompose the set of feasible $ (Y_{i,j}, \theta_{v_{i,j}}) $, denoted as $ S=S_{\theta} \times S_{Y}\subset \mathbb{R}^2 $, into the following subsets:
\begin{align}\label{eq:subsets}
&S_1=\bigcup_{\eta=1, 2}S_{1,\eta}=\{(Y_{i,j}, \theta_{v_{i,j}})| f_{i,j}>0, Y_{i,j}\frac{df_{i,j}}{d\theta_{v_{i,j}}}>0\},\notag \\
&S_2=\bigcup_{\eta=1, 2}S_{2,\eta}=\{(Y_{i,j}, \theta_{v_{i,j}})| h_{i,j}>0\},\notag \\
&S_3=\bigcup_{\eta=1, 2, 3, 4}S_{3,\eta}=\{(Y_{i,j},\theta_{v_{i,j}})| f_{i,j}, h_{i,j}\geq 0, Y_{i,j}\frac{df_{i,j}}{d\theta_{v_{i,j}}}\geq 0\},\notag \\
&S_0=S\setminus \big( \bigcup_{\eta=1, 2, 3} S_{\eta}\big),
\end{align}
which are illustrated in Figure \ref{fig:prf}. Since $ Y_{i,j} $, $ \theta_{v_{i,j}} $, $ f_{i,j} $ and $ h_{i,j} $ are continuous functions, the proposed projection operator (\ref{eq:17nx}) is continuous in each subspace of $ S $. Here, we will prove that the proposed projection is continuous also on the boundaries of these subsets.
\begin{figure}
	\begin{center}
		\vspace{0.2cm}
		\includegraphics[width=9cm]{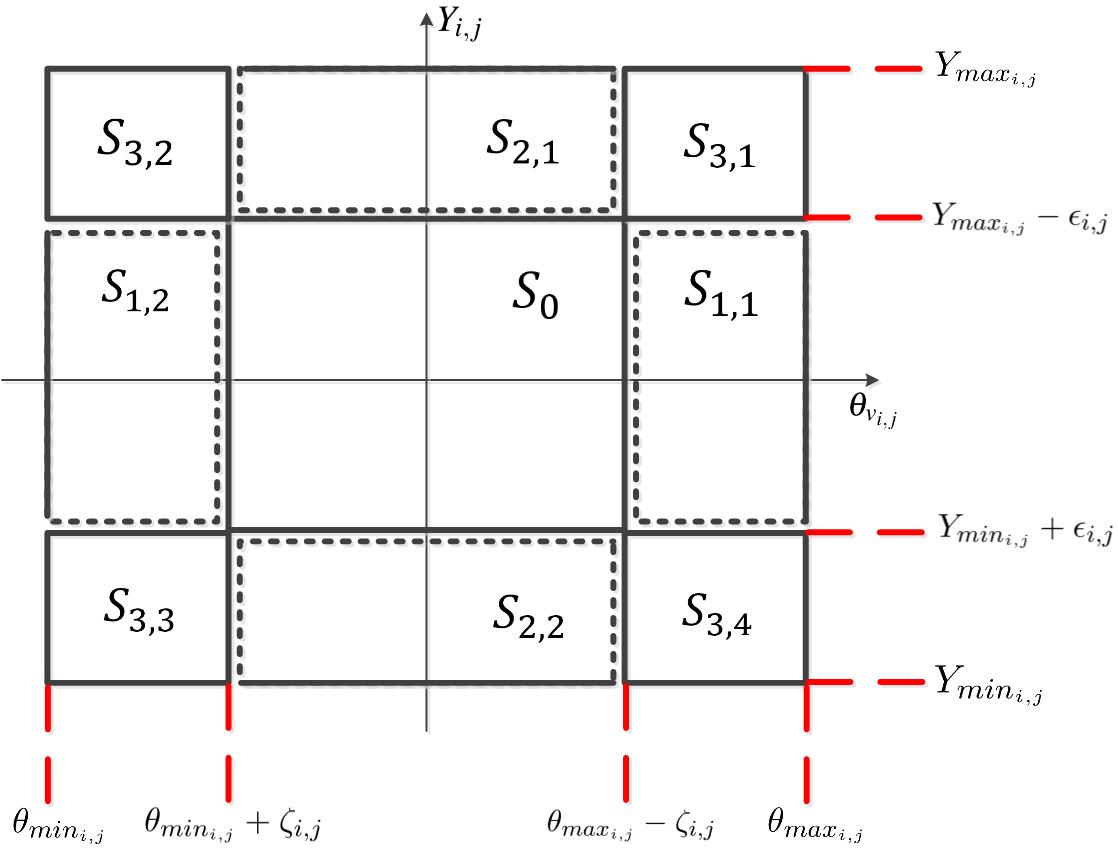}    %
		\caption{The decomposed set of feasible ($ Y_{i,j}, \theta_{v_{i,j}} $).} 
		\label{fig:prf}
	\end{center}
\end{figure}

 
Consider the boundary between $ S_0 $ and $ S_{2,1} $ (see Figure \ref{fig:prf}). Let the point $ (\theta_0,Y_{max_{i,j}}-\epsilon_{i,j})\in S_0 $ be an arbitrary point on the boundary. Notice that since $ S_0 $ is a closed set, the points on the boundary of $ S_0 $ and $ S_{2,1} $ belong to $ S_0 $. Therefore, in order to show that the proposed projection algorithm is continuous on the boundary of $ S_0 $ and $ S_{2,1} $, we should show that 
\begin{align}
\lim_{(\theta_{v_{i,j}}, Y_{i,j})\to (\theta_0, Y_{max_{i,j}}-\epsilon_{i,j})}\text{Proj}_m(\theta_{v_{i,j}}, Y_{i,j})=\text{Proj}_m(\theta_0,Y_{max_{i,j}}-\epsilon_{i,j})=Y_{max_{i,j}}-\epsilon_{i,j}
\end{align}
in both sets, $ S_0 $ and $ S_{2,1} $. 


First, consider taking the limit in the set $ S_2 $. For any given $ \gamma>0 $, there exists $ \delta_1=min\{\sqrt{2}\epsilon_{i,j},\frac{\sqrt{2}\epsilon_{i,j}\gamma}{\epsilon_{i,j}+Y_{max_{i,j}}}\} $ such that for $ Y_{i,j}\in (Y_{max_{i,j}}-\epsilon_{i,j},Y_{max_{i,j}}-\epsilon_{i,j}+\frac{\delta_1}{\sqrt{2}}) $ and $ \theta_{v_{i,j}}\in (\theta_0-\frac{\delta_1}{2\sqrt{2}},\theta_0+\frac{\delta_1}{2\sqrt{2}}) $, $0<\sqrt{(\theta_{i,j}-\theta_0)^2+(Y_{i,j}-Y_{max_{i,j}}+\epsilon_{i,j})^2}\leq \delta_1$. Then using $ |Y_{i,j}-Y_{max_{i,j}}+\epsilon_{i,j}|<\frac{\delta_1}{\sqrt{2}} $ we have
\begin{align}\label{eq:cont1}
|\text{Proj}_m(\theta_{v_{i,j}},Y_{i,j})-Y_{max_{i,j}}&+\epsilon_{i,j}|=|Y_{i,j}(1-\hat{h}_{i,j})-Y_{max_{i,j}}+\epsilon_{i,j}| \\
&\leq |Y_{i,j}-Y_{max_{i,j}}+\epsilon_{i,j}|+|Y_{i,j}\hat{h}_{i,j}|\notag \\
&<\frac{\delta_1}{\sqrt{2}}+\left| \frac{Y_{i,j}(Y_{i,j}-Y_{min_{i,j}}-\epsilon_{i,j})(Y_{i,j}-Y_{max_{i,j}}+\epsilon_{i,j})}{(Y_{max_{i,j}}-Y_{min_{i,j}}-\epsilon_{i,j})\epsilon_{i,j}}\right| .\notag
\end{align}
Considering $ Y_{i,j}\in (Y_{max_{i,j}}-\epsilon_{i,j},Y_{max_{i,j}}-\epsilon_{i,j}+\frac{\delta_1}{\sqrt{2}}) $, an upper bound on (\ref{eq:cont1}) can be calculated as
\begin{align}\label{eq:cont2}
&|\text{Proj}_m(\theta_{v_{i,j}},Y_{i,j})-Y_{max_{i,j}}+\epsilon_{i,j}| \notag \\
&<\frac{\delta_1}{\sqrt{2}}+\left| \frac{(Y_{max_{i,j}}-\epsilon_{i,j}+\frac{\delta_1}{\sqrt{2}})(Y_{max_{i,j}}-Y_{min_{i,j}}-2\epsilon_{i,j}+\frac{\delta_1}{\sqrt{2}})(\frac{\delta_1}{\sqrt{2}})}{(Y_{max_{i,j}}-Y_{min_{i,j}}-\epsilon_{i,j})\epsilon_{i,j}}\right|.
\end{align}
If $ \sqrt{2}\epsilon_{i,j}\leq \frac{\sqrt{2}\epsilon_{i,j}\gamma}{\epsilon_{i,j}+Y_{max_{i,j}}} $, then $ \gamma\geq \epsilon_{i,j}+Y_{max_{i,j}} $, and $ \delta_1=\sqrt{2}\epsilon_{i,j} $. Substituting $ \sqrt{2}\epsilon_{i,j} $ for $ \delta_1 $ in (\ref{eq:cont2}) leads to  
\begin{align}\label{eq:cont2x}
&|\text{Proj}_m(\theta_{v_{i,j}},Y_{i,j})-Y_{max_{i,j}}+\epsilon_{i,j}|
<\epsilon_{i,j}+Y_{max_{i,j}}\leq \gamma.
\end{align}
On the other hand, if $ \sqrt{2}\epsilon_{i,j}> \frac{\sqrt{2}\epsilon_{i,j}\gamma}{\epsilon_{i,j}+Y_{max_{i,j}}} $, then $ \gamma< \epsilon_{i,j}+Y_{max_{i,j}} $, and $ \delta_1=\frac{\sqrt{2}\epsilon_{i,j}\gamma}{\epsilon_{i,j}+Y_{max_{i,j}}} $. Substituting $ \frac{\sqrt{2}\epsilon_{i,j}\gamma}{\epsilon_{i,j}+Y_{max_{i,j}}} $ in (\ref{eq:cont2}) leads to  
\begin{align}\label{eq:cont2xx}
&|\text{Proj}_m(\theta_{v_{i,j}},Y_{i,j})-Y_{max_{i,j}}+\epsilon_{i,j}| <\frac{\epsilon_{i,j}\gamma}{\epsilon_{i,j}+Y_{max_{i,j}}}\notag \\
&+\left| \frac{(Y_{max_{i,j}}-\epsilon_{i,j}+\frac{\epsilon_{i,j}\gamma}{\epsilon_{i,j}+Y_{max_{i,j}}})(Y_{max_{i,j}}-Y_{min_{i,j}}-2\epsilon_{i,j}+\frac{\epsilon_{i,j}\gamma}{\epsilon_{i,j}+Y_{max_{i,j}}})(\frac{\epsilon_{i,j}\gamma}{\epsilon_{i,j}+Y_{max_{i,j}}})}{(Y_{max_{i,j}}-Y_{min_{i,j}}-\epsilon_{i,j})\epsilon_{i,j}}\right|.
\end{align}
Since $ Y_{max_{i,j}}-\epsilon_{i,j}>0 $ and $ Y_{min_{i,j}}+\epsilon_{i,j}<0 $, we have $ Y_{max_{i,j}}-Y_{min_{i,j}}-2\epsilon_{i,j}>0 $. Using these inequalities, and the fact that $ \gamma< \epsilon_{i,j}+Y_{max_{i,j}} $, (\ref{eq:cont2xx}) can be rewritten as
\begin{align}\label{eq:cont2xxxx}
&|\text{Proj}_m(\theta_{v_{i,j}},Y_{i,j})-Y_{max_{i,j}}+\epsilon_{i,j}|< \frac{(\epsilon_{i,j}+Y_{max_{i,j}})\gamma}{\epsilon_{i,j}+Y_{max_{i,j}}}= \gamma.
\end{align}
Therefore, $ \lim_{(\theta_{v_{i,j}}, Y_{i,j})\to (\theta_0, Y_{max_{i,j}}-\epsilon_{i,j})}\text{Proj}_m(\theta_{v_{i,j}}, Y_{i,j})=Y_{max_{i,j}}-\epsilon_{i,j} $ in set $ S_{2,1} $. 

Let us now consider the same limit operation in $ S_0 $. Again, for any $ \gamma>0 $, there exist a $ \delta_1=min\{\sqrt{2}\epsilon_{i,j},\frac{\sqrt{2}\epsilon_{i,j}\gamma}{\epsilon_{i,j}+Y_{max_{i,j}}}\} $ such that for $ Y_{i,j}\in (Y_{max_{i,j}}-\epsilon_{i,j}-\frac{\delta_1}{\sqrt{2}},Y_{max_{i,j}}-\epsilon_{i,j}) $ and $ \theta_{v_{i,j}}\in (\theta_0-\frac{\delta_1}{2\sqrt{2}},\theta_0+\frac{\delta_1}{2\sqrt{2}}) $, $0<\sqrt{(\theta_{i,j}-\theta_0)^2+(Y_{i,j}-Y_{max_{i,j}}+\epsilon_{i,j})^2}\leq \delta_1$. Then using $ |Y_{i,j}-Y_{max_{i,j}}+\epsilon_{i,j}|<\frac{\delta_1}{\sqrt{2}} $ we have
\begin{align}\label{eq:cont3}
|\text{Proj}_m(\theta_{v_{i,j}},Y_{i,j})-Y_{max_{i,j}}+\epsilon_{i,j}|&=|Y_{i,j}-Y_{max_{i,j}}+\epsilon_{i,j}| \notag \\
& <\frac{\delta_1}{\sqrt{2}}\leq \frac{\epsilon_{i,j}}{\epsilon_{i,j}+Y_{max_{i,j}}}\gamma<\gamma.
\end{align}
This shows that $ \lim_{(\theta_{v_{i,j}}, Y_{i,j})\to (\theta_0, Y_{max_{i,j}}-\epsilon_{i,j})}\text{Proj}_m(\theta_{v_{i,j}}, Y_{i,j})=Y_{max_{i,j}}-\epsilon_{i,j} $ in $ S_{0} $. Therefore, $ \text{Proj}_m(\theta_{v_{i,j}},Y_{i,j}) $ is continuous on the boundary of $ S_0 $ and $ S_{2,1} $.

		Consider now the boundary between $ S_{1,1} $ and  $ S_{3,1} $ (see Figure \ref{fig:prf}). Let the point $ (\theta_1,Y_{max_{i,j}}-\epsilon_{i,j}) $ be an arbitrary point on the boundary of $ S_{1,1} $ and $ S_{3,1} $. Notice that since $ S_{3,1} $ is a closed set, the points on the boundary of $ S_1 $ and $ S_3 $ belong to $ S_3 $. We should show that the limit of $ \text{Proj}_m(\theta_{v_{i,j}},Y_{i,j}) $ when $ (\theta_{v_{i,j}},Y_{i,j}) $ approaches $ (\theta_1,Y_{max_{i,j}}-\epsilon_{i,j}) $ in $ S_{3,1} $ leads to the same value as $ (\theta_{v_{i,j}},Y_{i,j}) $ approaches to $ (\theta_1,Y_{max_{i,j}}-\epsilon_{i,j}) $ in $ S_{1,1} $, and this value is equal to $ \text{Proj}_m(\theta_1,Y_{max_{i,j}}-\epsilon_{i,j})=(Y_{max_{i,j}}-\epsilon_{i,j})(1-\hat{h}(Y_{max_{i,j}}-\epsilon_{i,j}))(1-\hat{f}({\theta_1}))=(Y_{max_{i,j}}-\epsilon_{i,j})(1-\hat{f}({\theta_1})) $.
	
	First, consider the limit in $ S_{3,1} $. For any $ \gamma>0 $, there exists $ \delta_2=min\{\sqrt{2}\epsilon_{i,j},\sqrt{2}\gamma X^{-1}\} $, where $ X=1+\hat{f}(\theta_1)+(\frac{2\theta_1-\theta_{max_{i,j}}-\theta_{min_{i,j}}+\epsilon_{i,j}\slash 2}{2(\theta_{max_{i,j}}-\theta_{min_{i,j}}-\zeta_{i,j})\zeta_{i,j}})Y_{max_{i,j}} $, such that
	 for $ Y_{i,j}\in (Y_{max_{i,j}}-\epsilon_{i,j},Y_{max_{i,j}}-\epsilon_{i,j}+\frac{\delta_2}{\sqrt{2}}) $ and $ \theta_{v_{i,j}}\in (\theta_1-\frac{\delta_2}{2\sqrt{2}},\theta_1+\frac{\delta_2}{2\sqrt{2}}) $, $0<\sqrt{(\theta_{i,j}-\theta_0)^2+(Y_{i,j}-Y_{max_{i,j}}+\epsilon_{i,j})^2}\leq \delta_2$. Then, we have
	\begin{align}\label{eq:cont4}
	&|\text{Proj}_m(\theta_{v_{i,j}},Y_{i,j})-(Y_{max_{i,j}}-\epsilon_{i,j})(1-\hat{f}(\theta_1))|\notag\\ &=|Y_{i,j}(1-\hat{f}_{i,j})(1-\hat{h}_{i,j})-(Y_{max_{i,j}}-\epsilon_{i,j})(1-\hat{f}(\theta_1))| \notag \\
	&\leq|Y_{i,j}(1-\hat{f}(\theta_1-\frac{\delta_2}{2\sqrt{2}}))-(Y_{max_{i,j}}-\epsilon_{i,j})(1-\hat{f}(\theta_1+\frac{\delta_2}{2\sqrt{2}}))|\notag \\
	&\leq |Y_{i,j}-Y_{max_{i,j}}+\epsilon_{i,j}|+|Y_{i,j}\hat{f}(\theta_1-\frac{\delta_2}{2\sqrt{2}})-(Y_{max_{i,j}}-\epsilon_{i,j})\hat{f}(\theta_1+\frac{\delta_2}{2\sqrt{2}})|\notag \\
	&< \frac{\delta_2}{\sqrt{2}}+|(Y_{max_{i,j}}-\epsilon_{i,j}+\frac{\delta_2}{\sqrt{2}})\hat{f}(\theta_1-\frac{\delta_2}{2\sqrt{2}})-(Y_{max_{i,j}}-\epsilon_{i,j})\hat{f}(\theta_1+\frac{\delta_2}{2\sqrt{2}})|.
	\end{align}
	It can be shown that $ \hat{f}(\theta_1-\frac{\delta_2}{2\sqrt{2}})=\hat{f}(\theta_1)-\frac{\delta_2}{\sqrt{2}}(\frac{2\theta_1-\theta_{max_{i,j}}-\theta_{min_{i,j}}+\frac{\delta_2}{2\sqrt{2}}}{2(\theta_{max_{i,j}}-\theta_{min_{i,j}}-\zeta_{i,j})\zeta_{i,j}}) $ and $ \hat{f}(\theta_1+\frac{\delta_2}{2\sqrt{2}})=\hat{f}(\theta_1)+\frac{\delta_2}{\sqrt{2}}(\frac{2\theta_1-\theta_{max_{i,j}}-\theta_{min_{i,j}}+\frac{\delta_2}{2\sqrt{2}}}{2(\theta_{max_{i,j}}-\theta_{min_{i,j}}-\zeta_{i,j})\zeta_{i,j}}) $. Therefore, an upper bound on (\ref{eq:cont4}) can be obtained as
	\begin{align}\label{eq:cont4x}
	&|\text{Proj}_m(\theta_{v_{i,j}},Y_{i,j})-(Y_{max_{i,j}}-\epsilon_{i,j})(1-\hat{h}(Y_{max_{i,j}}-\epsilon_{i,j}))(1-\hat{f}(\theta_1))|\notag\\
	&< \frac{\delta_2}{\sqrt{2}}+\frac{\delta_2}{\sqrt{2}}\hat{f}(\theta_1)+\frac{\delta_2}{\sqrt{2}}(\frac{2\theta_1-\theta_{max_{i,j}}-\theta_{min_{i,j}}+\frac{\delta_2}{2\sqrt{2}}}{2(\theta_{max_{i,j}}-\theta_{min_{i,j}}-\zeta_{i,j})\zeta_{i,j}})(Y_{max_{i,j}}-\epsilon_{i,j}+\frac{\delta_2}{\sqrt{2}}).
	\end{align}
	Using the definition of $ \delta_2 $, and the fact that $\theta_{max_{i,j}}-\zeta_{i,j}>0$ and $\theta_{min_{i,j}}+\zeta_{i,j}<0$, an upper bound on (\ref{eq:cont4x}) can be obtained as
	\begin{align}\label{eq:cont4xx}
	&|\text{Proj}_m(\theta_{v_{i,j}},Y_{i,j})-(Y_{max_{i,j}}-\epsilon_{i,j})(1-\hat{h}(Y_{max_{i,j}}-\epsilon_{i,j}))(1-\hat{f}(\theta_1))|\notag\\
	&< \frac{\delta_2}{\sqrt{2}}((1+\hat{f}(\theta_1)+(\frac{2\theta_1-\theta_{max_{i,j}}-\theta_{min_{i,j}}+\frac{\epsilon_{i,j}}{2}}{2(\theta_{max_{i,j}}-\theta_{min_{i,j}}-\zeta_{i,j})\zeta_{i,j}})Y_{max_{i,j}})\leq \gamma.
	\end{align}
	This shows that $ \lim_{(\theta_{v_{i,j}}, Y_{i,j})\to (\theta_1, Y_{max_{i,j}}-\epsilon_{i,j})}\text{Proj}_m(\theta_{v_{i,j}}, Y_{i,j})=(Y_{max_{i,j}}-\epsilon_{i,j})(1-\hat{f}(\theta_1)) $, in set $ S_{3,1} $.
	
	Now, consider taking the same limit in $ S_{1,1} $. For $ Y_{i,j}\in (Y_{max_{i,j}}-\epsilon_{i,j}-\frac{\delta_2}{\sqrt{2}},Y_{max_{i,j}}-\epsilon_{i,j}) $ and $ \theta_{v_{i,j}}\in (\theta_1-\frac{\delta_2}{2\sqrt{2}},\theta_1+\frac{\delta_2}{2\sqrt{2}}) $, $0<\sqrt{(\theta_{i,j}-\theta_0)^2+(Y_{i,j}-Y_{max_{i,j}}+\epsilon_{i,j})^2}\leq \delta_2$. Then, we have
	\begin{align}\label{eq:cont5}
	&|\text{Proj}_m(\theta_{v_{i,j}},Y_{i,j})-(Y_{max_{i,j}}-\epsilon_{i,j})(1-\hat{h}(Y_{max_{i,j}}-\epsilon_{i,j}))(1-\hat{f}(\theta_1))|\notag \\
	&=|Y_{i,j}(1-\hat{f}_{i,j})-(Y_{max_{i,j}}-\epsilon_{i,j})(1-\hat{f}(\theta_1))|\notag \\
	&\leq|Y_{i,j}(1-\hat{f}(\theta_1-\frac{\delta_2}{2\sqrt{2}}))-(Y_{max_{i,j}}-\epsilon_{i,j})(1-\hat{f}(\theta_1+\frac{\delta_2}{2\sqrt{2}}))|.
	\end{align}
	Using the same procedure as (\ref{eq:cont4})-(\ref{eq:cont4xx}), it can be shown that $ |\text{Proj}_m(\theta_{v_{i,j}},Y_{i,j})-(Y_{max_{i,j}}-\epsilon_{i,j})(1-\hat{h}(Y_{max_{i,j}}-\epsilon_{i,j}))(1-\hat{f}(\theta_1))|<\gamma $. Therefore, $ \text{Proj}_m(\theta_{v_{i,j}},Y_{i,j}) $ is continuous on the boundary of $ S_{1,1} $ and $ S_{3,1} $.
	
	Continuity of the proposed projection function on the other boundaries can be proved following the same procedure as above. Therefore, $ \text{Proj}_m(\theta_{v_{i,j}},Y_{i,j}) $ is continuous on $ S $.
\end{proof}

 The final step before presenting the main theorem of this study is showing that the solution of the differential equation providing the parameter adaptation law $ \dot{\theta}_{v_{i,j}}=\text{Proj}_m({\theta}_{v_{i,j}}, Y_{i,j}) $, actually exists and is unique. Considering that $ \theta_{v_{i,j}} $ and $ Y_{i,j} $ are piecewise continuous functions of time, it is enough to prove that $ \text{Proj}_m(\theta_{v_{i,j}}, Y_{i,j}) $ is locally Lipschitz to show existence and uniqueness.
%
 
\begin{lemma}\label{lem7}
	The function $\emph{Proj}_m({\theta}_{v_{i,j}}, Y_{i,j}):S_{\theta} \times S_{Y} \rightarrow \mathbb{R}$, where $ S_{\theta}, S_{Y}\subset \mathbb{R} $, is locally Lipschitz.
\end{lemma}
\begin{proof}

	In order to prove that a function $ g:D\subset \mathbb{R}^n \rightarrow \mathbb{R}^m $  is locally Lipschitz, it must be shown that there exists a positive constant $ K $ such that $ ||g(x)-g(y)||\leq K||x-y|| $, for any $ x,y \in D\subset \mathbb{R}^n $. Let $ a^1\equiv ( Y_{i,j}^1, \theta_{v_{i,j}}^1)\in S\subset \mathbb{R}^2 $ and $ a^0\equiv ( Y_{i,j}^0, \theta_{v_{i,j}}^0)\in S\subset \mathbb{R}^2 $, where $ S $ is given as $ S=S_{\theta} \times S_{Y}\subset \mathbb{R}^2 $. Furthermore, let $ a^{\mu}=( Y_{i,j}^{\mu}, \theta_{v_{i,j}}^{\mu}) $, $  \mu \in [0,1] $, be any point on the line connecting $ a^0 $ and $ a^1 $, which satisfy
	\begin{align}
	&Y_{i,j}^{\mu}=\mu Y_{i,j}^1+(1-\mu)Y_{i,j}^0,\label{eq:e1}\\
	&\theta_{i,j}^{\mu}=\mu \theta_{v_{i,j}}^1+(1-\mu)\theta_{v_{i,j}}^0.\label{eq:e2}
	\end{align}
	The Lipschitz condition needs to be investigated for $ 4 $ different cases, which are given below. The subsets of $ S $, defined in (\ref{eq:subsets}) and demonstrated in Figure \ref{fig:prf}, are used throughout the proof.
	
	Case 1: If for all $ \mu\in [0, 1] $, $ a^{\mu} $ lies in the set $ S_0 $, then, using (\ref{eq:17nx}), it can be shown that
	\begin{align}
	|\text{Proj}_m(a^1)-\text{Proj}_m(a^0)|&=|Y_{i,j}^1-Y_{i,j}^0|\notag\\
	&\leq |Y_{i,j}^1-Y_{i,j}^0|+|\theta_{v_{i,j}}^1-\theta_{v_{i,j}}^0|\notag \\
	&\leq k_0||a^1-a^0||,
	\end{align}
	where $ k_0 $ is a positive constant. This satisfies the Lipschitz condition on $ S_0 $.
	
	Case 2: If for all $ \mu\in [0, 1] $, $ a^{\mu} $ lies in the set $ S_{3,1} $, then
	\begin{align}\label{eq:e1xx}
	|\text{Proj}_m(a^1)-\text{Proj}_m(a^0)|&=|Y_{i,j}^1(1-\hat{f}_{i,j}^1)(1-\hat{h}_{i,j}^1)\notag \\
	&-Y_{i,j}^0(1-\hat{f}_{i,j}^0)(1-\hat{h}_{i,j}^0)|,
	\end{align}
	where $ \hat{f}_{i,j}^{\ell}=\hat{f}(\theta_{v_{i,j}}^{\ell}) $ and $ \hat{h}_{i,j}^{\ell}=\hat{h}(Y_{{i,j}}^{\ell}) $ for $ \ell=\{0, 1\} $. Using (\ref{eq:18n}) and (\ref{eq:18nxx}), it can be shown that there exist positive constants $ k_{\theta 0} $ and $ k_{Y0} $ such that
	\begin{align}
	|\hat{f}_{i,j}^1-\hat{f}_{i,j}^0|&<k_{\theta 0}|\theta_{v_{i,j}}^1-\theta_{v_{i,j}}^0|\label{eq:e11} \\
	|\hat{h}_{i,j}^1-\hat{h}_{i,j}^0|&<k_{Y0}|Y_{{i,j}}^1-Y_{{i,j}}^0|.\label{eq:e12}
	\end{align}
	Using (\ref{eq:e11}) and (\ref{eq:e12}), an upper bound on (\ref{eq:e1xx}) can be obtained as
	\begin{align}
	|\text{Proj}_m(a^1)-\text{Proj}_m(a^0)|&\leq k_{Y1}|Y_{i,j}^1-Y_{i,j}^0|+k_{\theta 1}|\theta_{v_{i,j}}^1-\theta_{v_{i,j}}^0|\notag \\
	&\leq k_1||a^1-a^0||,
	\end{align}
	where $ k_{\theta 1} $, $ k_{Y1} $ and $ k_1 $ are positive constants. The same procedure can be followed for each subsets of $ S_{1} $, $ S_{2} $ and $ S_3 $, and therefore the Lipschitz condition is satisfied on each subsets of $ S_1, S_2 $ and $ S_3 $. 
	

	Case 3: If $ a^0 $ and $ a^1 $ are in two neighboring subsets of $ S $, then the following analysis can be conducted: Let $ a^1 $ belong to $ S_{3,1} $ and $ a^0 $ to $ S_{1,1} $. Then, the segment $ [ a^0, a^1 ] $ can be divided into two segments $ [a^1, a^{\mu^*}]\in S_{3,1} $ and $ (a^{\mu^*}, a^0]\in S_{1,1} $, where
	\begin{align}\label{eq:e4}
	\mu^*&=\min \ \mu,\notag \\
	\text{s.t.}\ \ \ \mu&\in[0, 1] \ \text{and}\ 
	a^{\mu}\in S_{3,1}.
	\end{align}
	Using the mean value theorem in $ S_{1,1}\setminus \partial S_{1,1}  $, where $ \partial S_{1,1} $ denotes the boundary of $ S_{1,1} $, and using (\ref{eq:e1}) and (\ref{eq:e2}), we obtain that
	\begin{align}\label{eq:e3}
	|\text{Proj}_m(a^{\mu^*})-\text{Proj}_m(a^0)|&\leq k_2^{'}(|Y_{i,j}^{\mu^*}-Y_{i,j}^0|+|\theta_{v_{i,j}}^{\mu^*}-\theta_{v_{i,j}}^0|)\notag \\
	&\leq k_2^{''}(|Y_{i,j}^{1}-Y_{i,j}^0|+|\theta_{v_{i,j}}^{1}-\theta_{v_{i,j}}^0|), 
	%
	\end{align}
	where $ k_2^{'} $ and $ k_2^{''} $ are positive constants. Also, following the procedure in Case 2, it can be shown that $ |\text{Proj}_m(a^1)-\text{Proj}_m(a^{\mu^*})|\leq k_1||a^1-a^0|| $. Therefore, using the triangle inequality, we get
	\begin{align}
	|\text{Proj}_m(a^1)-\text{Proj}_m(a^0)|&\leq |\text{Proj}_m(a^1)-\text{Proj}_m(a^{\mu^*})|\notag \\
	&+|\text{Proj}_m(a^{\mu^*})-\text{Proj}_m(a^0)|\notag \\
	&\leq k_2||a^1-a^0||,
	\end{align}
	where $ k_2 $ is a positive constant. The same procedure can be used for the other two neighboring subsets.
	
	Case 4: If $ a^0 $ and $ a^1 $ are in two non-neighboring subsets of $ S $, then the following analysis can be conducted: Let $ a^0 $ belong to $ S_{1,1} $ and $ a^1 $ to $ S_{2,1} $. Then, the segment $ [ a^0, a^1 ] $ can be divided into three segments $ [a^0, a^{\alpha^*})\in S_{1,1} $, $ [a^{\alpha^*}, a^{\beta^*}]\in S_0 $, and $ (a^{\beta^*}, a^{a^1}]\in S_{2,1} $, where $ \alpha^* $ and $ \beta^* $ are defined as
	\begin{align}
	\alpha^*&=\min \ \mu\notag \\
	\text{s.t.}\ \ \ \mu&\in[0, 1]\ \text{and}\ 
	a^{\mu}\in S_0,
	\end{align}
	and
	\begin{align}
	\beta^*&=\max \ \mu\notag \\
	\text{s.t.}\ \ \ \mu&\in[0, 1]\ \text{and}\
	a^{\mu}\in S_0.
	\end{align}
	Then, the same procedure used in Case 3 can be followed to obtain the Lipschitz condition.
	
	Since the Lipschitz condition is satisfied for any two points $ a^0, a^1\in S $, the projection algorithm is locally Lipschitz on $ S $.
\end{proof}

After defining the modified projection algorithm, proving its properties that will be useful in the stability analysis of the closed loop system, and proving the existence and uniqueness of the solution of the differential equation describing the algorithm, we provide the main theorem below, stating that when the proposed projection algorithm is employed, all the signals in the adaptive control allocation system, in the presence of actuator magnitude and rate saturation, remains bounded and the control allocation error converges to a predetermined closed set. 
\begin{theorem}\label{thm2}
	Consider the actuator command signal $ u $ produced by the adaptive control allocation  (\ref{eq:aloc}) with $ g(\theta_{v}, Y(v_s, e))=\Gamma\text{Proj}_m(\theta_{v}, Y(v_s,e)) $, where $ \Gamma $ is a diagonal positive definite matrix and the projection operator is defined in (\ref{eq:17nx}) with convex functions (\ref{eq:18n}) and (\ref{eq:18nxx}).
	If $ Y=-v_se^TPB $, where $P$ is the positive definite symmetric matrix solution of the Lyapunov equation $ A_m^{T}P + PA_m = -Q $ with a symmetric positive definite matrix $Q$, then $ \tilde{\theta}_v $ and $ e $ remain bounded and converge to the compact set
	\begin{align} 
	E_2=\{ (e,\tilde{\theta}_v): ||e||^2\leq  \frac{2||\tilde{\theta}_v||_F^2||Y_{MAX}||_F}{\lambda_{min}(Q)}, ||\tilde{\theta}||\leq \tilde{\theta}_{max} \}.
	\end{align}
	Moreover, the design parameters $ \theta_{min_{i,j}} $, $ \theta_{max_{i,j}} $, $ Y_{min_{i,j}} $ and $ Y_{max_{i,j}} $ in (\ref{eq:18n}) and (\ref{eq:18nxx}) can be chosen such that for $ v_s\in \Omega_v=\{ v| -M_i\leq v_i\leq  M_i, -L_i\leq \dot{v}_i \leq L_i, i=1,..., r \} $, where $ M_i $ and $ L_i $ are positive scalars for $ i=1,..., r $, $ u $ remains in $ \Omega_u=\{u| u_{\text{min}_{j}}\leq u_j\leq u_{\text{max}_{j}}, \bar{u}_{\text{min}_{j}}\leq \dot{u}_j\leq \bar{u}_{\text{max}_{j}}, j=1,...,m \} $, where $ u_{min_j} $, $ u_{max_j} $, $ \bar{u}_{min_j} $, $ \bar{u}_{max_j} $ are actuator magnitude and rate constraints.
\end{theorem}
\begin{proof}
Substituting (\ref{eq:6}) into (\ref{eq:4}), we obtain that
\begin{equation} \label{eq:7}
\dot{\xi}=A_m\xi+(B\Lambda{\theta}_v^T-I)v_s.
\end{equation}
It is assumed that there exists an ideal adaptive parameter, $\theta_v^*$, such that
\begin{equation} \label{eq:8n}
B\Lambda{\theta}_v^{*T}=I.
\end{equation}
Since $ B\Lambda $ is a full row rank matrix, this assumption is always valid.
Defining $\theta_v^T=\theta_v^{*T}+\tilde{\theta}_v^T$, where $\tilde{\theta}_v^T$ is the deviation of $\theta_v^T$ from its ideal value, (\ref{eq:7}) can be rewritten as
\begin{equation} \label{eq:8}
\dot{\xi}=A_m\xi+B\Lambda \tilde{\theta}_v^Tv_s.
\end{equation}
Using (\ref{eq:5}) and (\ref{eq:8}), the error dynamics is obtained as
\begin{equation} \label{eq:9}
\dot{e}=A_me+B\Lambda \tilde{\theta}_v^Tv_s.
\end{equation}
Consider a Lyapunov function candidate
\begin{equation} \label{eq:10}
V=e^TPe+tr(\tilde{\theta}_v^T\Gamma^{-1}\tilde{\theta}_v\Lambda).
\end{equation}
The derivative of $ V $ along the trajectories of (\ref{eq:aloc}) can be calculated as
\begin{align} \label{eq:12}
\dot{V}=&e^T(A_m^TP + PA_m)e + 2e^TPB\Lambda\tilde{\theta}_v^Tv_s+ 2tr(\tilde{\theta}_v^T\Gamma^{-1}\dot{\tilde{\theta}}_v\Lambda)\notag \\
=&-e^TQe + 2e^TPB\Lambda\tilde{\theta}_v^Tv_s + 2tr(\tilde{\theta}_v^T\Gamma^{-1}\dot{\tilde{\theta}}_v\Lambda).
\end{align}
Using the property of the trace operation $a^Tb=tr(ba^T)$ where $a$ and $b$ are vectors, (\ref{eq:12}) can be rewritten as
\begin{equation} \label{eq:13}
\dot{V}=-e^TQe+2tr(\tilde{\theta}_v^T(v_se^TPB + \Gamma^{-1}\dot{\tilde{\theta}}_v)\Lambda).
\end{equation}
Substituting modified adaptive control law (\ref{eq:55}) into (\ref{eq:13}), the derivative of the Lyapunov function candidate is obtained as
	\begin{align} 
	&\dot{V}=-e^TQe+2tr(\tilde{\theta}_v^T(v_se^TPB + \text{Proj}_m(\theta_v, -v_se^TPB))\Lambda).
	\end{align}
	By using Lemma \ref{lem6}, we get
	\begin{align} 
	&\dot{V}\leq -\lambda_{min}(Q)||e||^T+2||\tilde{\theta}_v||_F^2||Y_{MAX}||_F,
	\end{align}
	where $ \lambda_{min}(\cdot) $ denotes the minimum eigenvalue. $ \dot{V}\leq 0 $ for $ ||e||^2\geq  (2||\tilde{\theta}_v||_F^2||Y_{MAX}||_F)/(\lambda_{min}(Q)) $.
	Therefore, for any initial conditions $ e(0) $ and $ \tilde{\theta}_{v}(0) $, if $ ||\tilde{\theta}_{v}(0)||\leq \tilde{\theta}_{max} $, where $ \tilde{\theta}_{max} $ is the predetermined upper bound for $ \tilde{\theta}_v $, $ e(t) $ and $ \tilde{\theta}_{v}(t) $ are bounded for all $ t \geq 0 $ and their trajectories converge to the following compact set (\cite{NarAnn12}),
	\begin{align} 
	E_2=\{ (e,\tilde{\theta}_v): ||e||^2\leq  \frac{2||\tilde{\theta}_v||_F^2||Y_{MAX}||_F}{\lambda_{min}(Q)}, ||\tilde{\theta}||\leq \tilde{\theta}_{max} \}.
	\end{align}
	Using Lemma \ref{lem5}, if the initial conditions are defined as $\theta_{v_{i,j}}(0)\in \Omega_{i,j}=\{\theta_{v_{i,j}}\in \mathbb{R}|f(\theta_{v_{i,j}})\leq 1\}$ and $Y_{{i,j}}(0)\in \bar{\Omega}_{i,j}=\{Y_{{i,j}}\in \mathbb{R}|
	h(Y_{{i,j}})\leq 1\}$, then $\theta_{v_{i,j}}(t)\in \Omega_{i,j}$ and $Y_{{i,j}}(t)\in \bar{\Omega}_{i,j}$ for all $ t\geq 0$.
	For a bounded $ v_s\in \Omega_v $, suitable values of $ \theta_{max_{i,j}} $, $ \theta_{min_{i,j}} $, $ Y_{max_{i,j}} $ and $ Y_{min_{i,j}} $ can be found to be used in $ f(\theta_{i,j}) $ and $ h(Y_{i,j}) $ that ensure $u_j\in [u_{\text{min}_{j}},u_{\text{max}_{j}}]$ and $\dot{u}_j\in [\bar{u}_{\text{min}_{j}},\bar{u}_{\text{max}_{j}}]$, $ j=1, ..., m $ for all $ t\geq 0 $. 
\end{proof}
\begin{remark}
	\label{rem0xxx}{It should be noted that \textit{control allocation}'s task is to distribute the total control effort produced by a \textit{controller} among redundant actuators. The investigated control allocation method and the proposed projection algorithm in this paper can be used with various different types of controllers. In this paper, a new control method is not proposed.}
\end{remark}

\begin{remark}
	\label{rem0xx}{Although the employment of the proposed projection algorithm is exemplified on an adaptive control allocation implementation, the proposed method can be extended to be used for other adaptive systems where the actuators are both magnitude and rate saturated.}
\end{remark}


\section{Application example}\label{sec6}
\subsection{ADMIRE model}
The Aerodata Model in Research Environment (ADMIRE) \citep{Har02}, which is an over-actuated aircraft model, is used for the simulations. The linearized model is given as 
\begin{equation}\label{eq:e59}
\begin{array}{ll}
\dot{x}=Ax+B_u u=Ax+B_vv_s,\\
v_s=Bu,\ \ \ B_u=B_vB,\ \ \  B_v=[0_{3\times2}\ I_{3\times3}]^T,\\
x=[\alpha \ \beta \ p \ q \ r]^T,\\
y=[p \ q \ r]^T,\\
u=[u_c \ u_{re} \ u_{le} \ u_r]^T,
\end{array}
\end{equation}
where $ \alpha,\ \beta,\ p,\ q $ and $ r $ are the angle of attack, sideslip angle, roll rate, pitch rate and yaw rate, respectively. The vector $u$ includes the commanded control surfaces' deflection. The control surfaces $ u_c,\ u_{re},\  u_{le}$ and $ u_r $ are the canard wings, right and left elevons and the rudder, respectively. The magnitude and rate limits of the commanded control surfaces are given as $ u_c\in[-55,25]\times \frac{\pi}{180}(rad)  ,u_{re},u_{le},u_r\in[-30,30]\times \frac{\pi}{180}(rad) $ and $ \dot{u}_c,\dot{u}_{re},\dot{u}_{le},\dot{u}_r\in[-40,40]\times \frac{\pi}{180}(rad/sec) $. The state and control matrices which are provided by \cite{Har02}, are given as
\begin{equation}\label{eq:e61}
\begin{array}{ll}
A= \left[ \begin{array}{cccccc}-0.5432 & 0.0137 & 0 & 0.9778 & 0 \\ 0 & -0.1179 & 0.2215 & 0 & -0.9661 \\  0 & -10.5123 & - 0.9967 & 0 & 0.6176 \\  2.6221 & -0.0030 & 0 & -0.5057 & 0 \\  0 & 0.7075 & -0.0939 & 0 & -0.2127\end{array} \right],\notag
\end{array}
\end{equation}
\vspace{-0.5cm}
\begin{equation}\label{eq:e62}
\begin{array}{ll}
B=\left[ \begin{array}{cccccccccccc}0 & -4.2423 & 4.2423 & 1.4871 \\ 1.6532 & -1.2735 & -1.2735 & 0.0024 \\  0 & -0.2805 & 0.2805 & -0.8823\end{array} \right].
\end{array}
\end{equation}

To introduce the actuator effectiveness uncertainty, we modify the model (\ref{eq:e59}) as
\begin{align}\label{eq:e59x}
\dot{x}&=Ax+B_u\Lambda u\notag \\
&=Ax+B_vB\Lambda u\notag \\
&=Ax+B_vv_s,
\end{align}
where $\Lambda \in \mathbb{R}^{4\times 4}$ is a diagonal matrix with uncertain positive elements. Substituting the allocated signal $ u $ given by (\ref{eq:6}), and using $\theta_v^T=\theta_v^{*T}+\tilde{\theta}_v^T$, (\ref{eq:e59x}) can be rewritten as
\begin{equation}\label{eq:e43xxx}
\begin{aligned}
\dot{x}&=Ax+B_vB\Lambda \theta_v^Tv_s=Ax+B_v(I+B\Lambda \tilde{\theta}_v^T)v_s,
\end{aligned} 
\end{equation}
where the total control input (see Figure 1) $ v\in \mathbb{R}^r $ can be designed using a proper control method. For the simulations conducted in this paper, we use the controller provided by \cite{TohYil19, TohYil20}.

\begin{figure}
	\begin{center}
		\includegraphics[width=14cm]{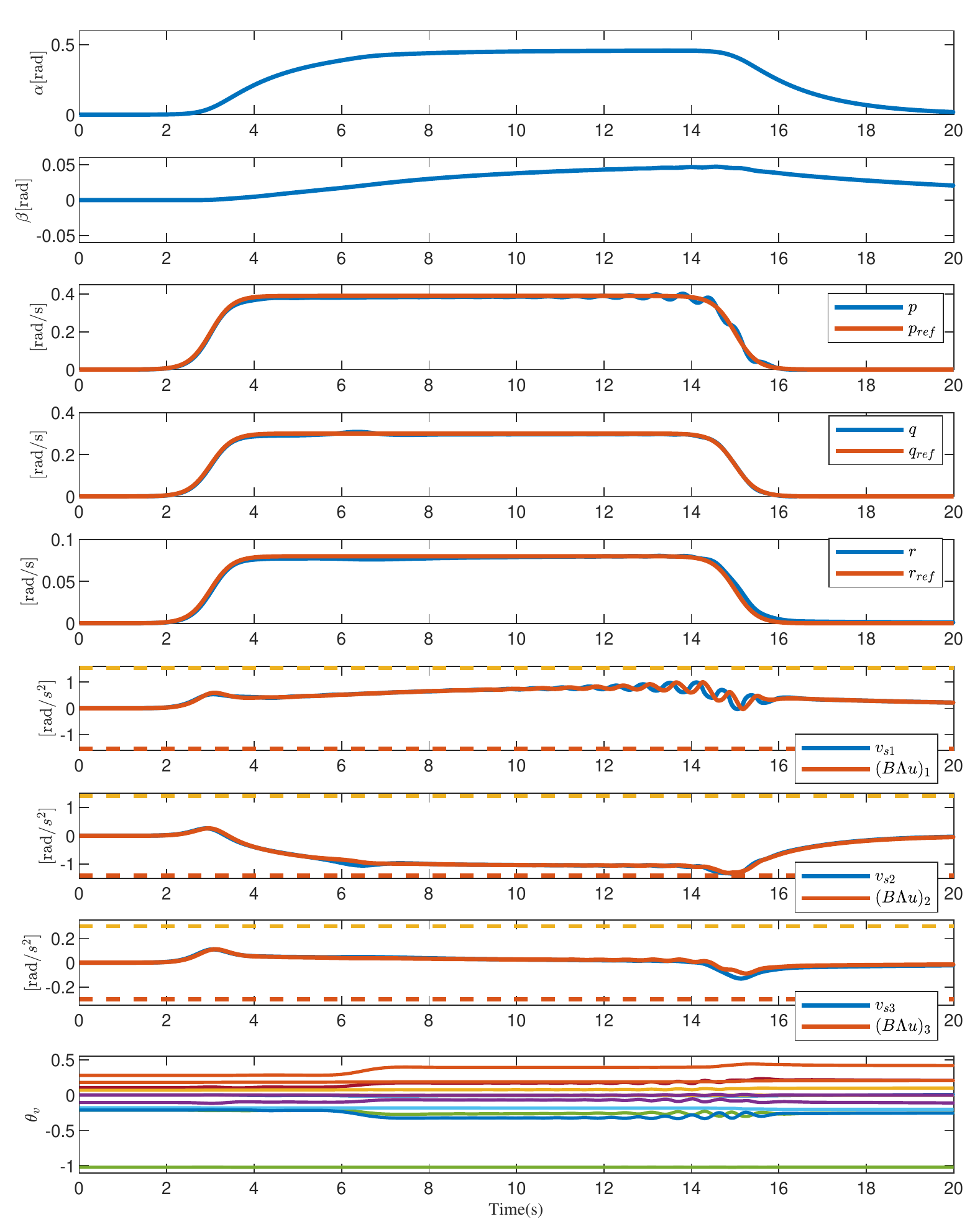}   
		\caption{Case I: Evolution of the states, total control inputs and adaptive parameters in the presence of \textbf{magnitude saturation}, using the \textbf{conventional projection method}.} 
		\label{fig:sys11}
	\end{center}
\end{figure} 
\begin{figure}
	\begin{center}
		\includegraphics[width=14cm]{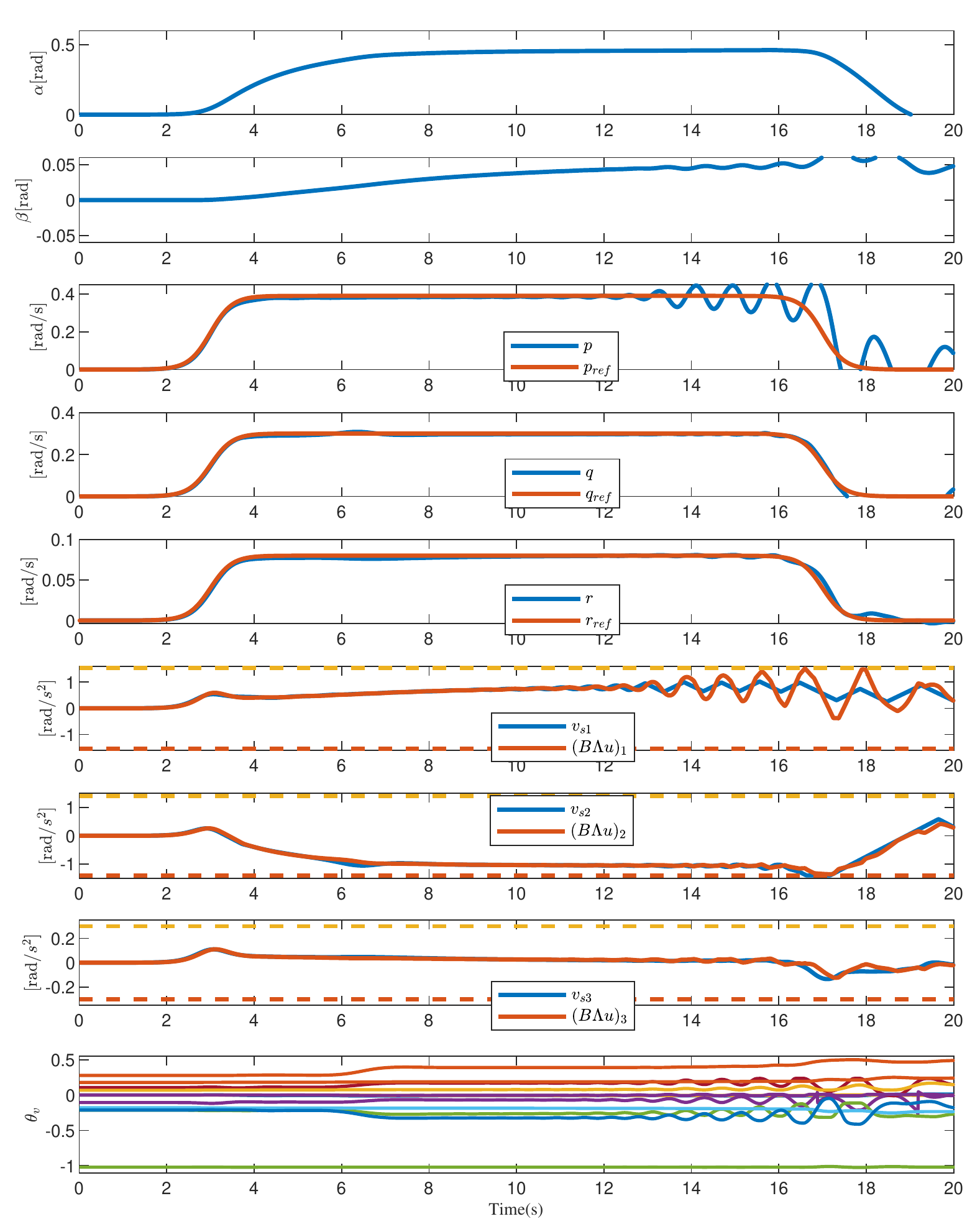}   
		\caption{Case II: Evolution of the states, total control inputs and adaptive parameters in the presence of \textbf{both magnitude and rate saturation}, using the \textbf{conventional projection method}.} 
		\label{fig:sys22}
	\end{center}
\end{figure} 
\begin{figure}
	\begin{center}
		\includegraphics[width=14cm]{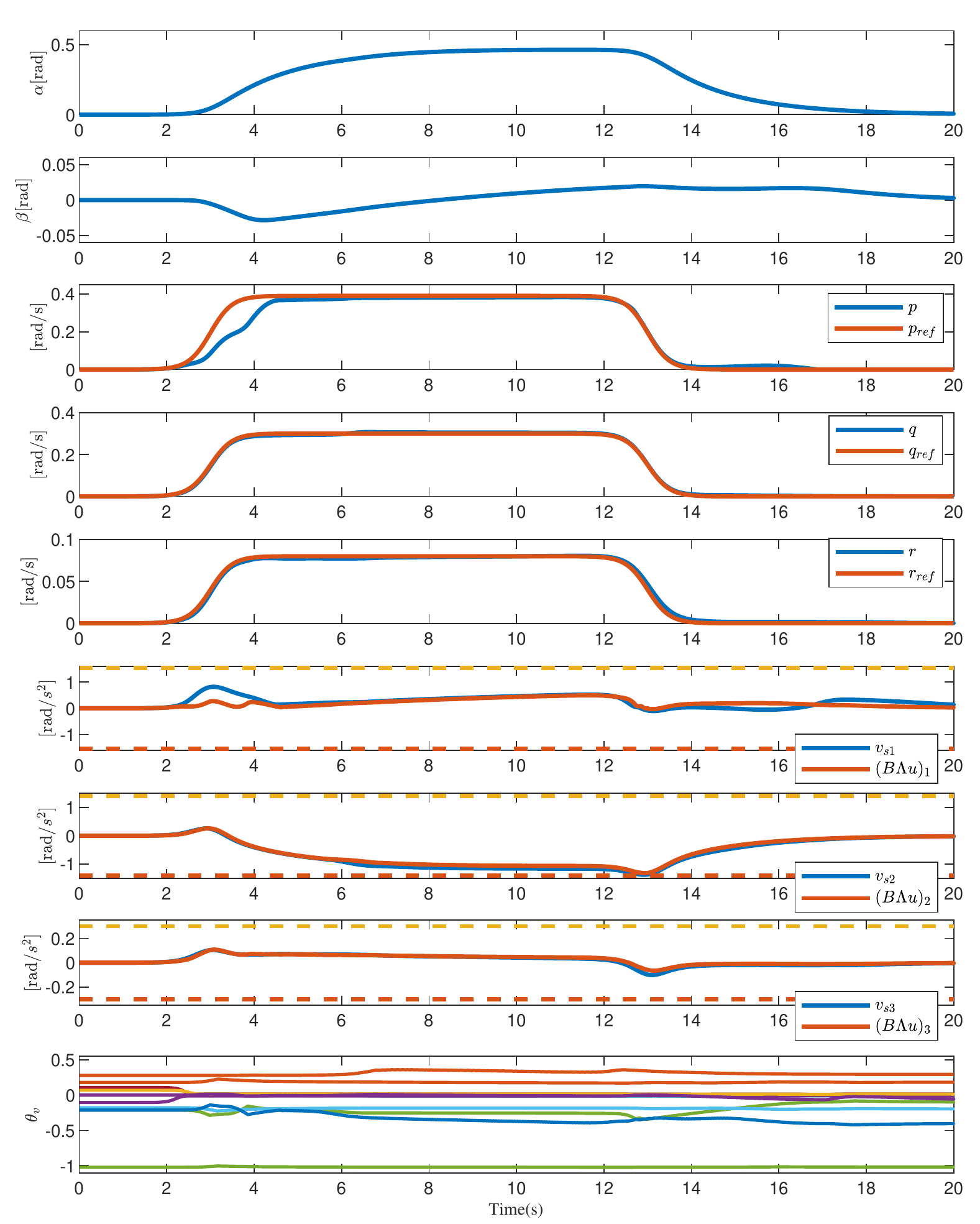}    %
		\caption{Case III: Evolution of the states, total control inputs and adaptive parameters in the presence of \textbf{both magnitude and rate saturation}, using the \textbf{proposed projection algorithm}.} 
		\label{fig:sys33}
	\end{center}
\end{figure}
\begin{figure}
	\begin{center}
		\includegraphics[width=14cm]{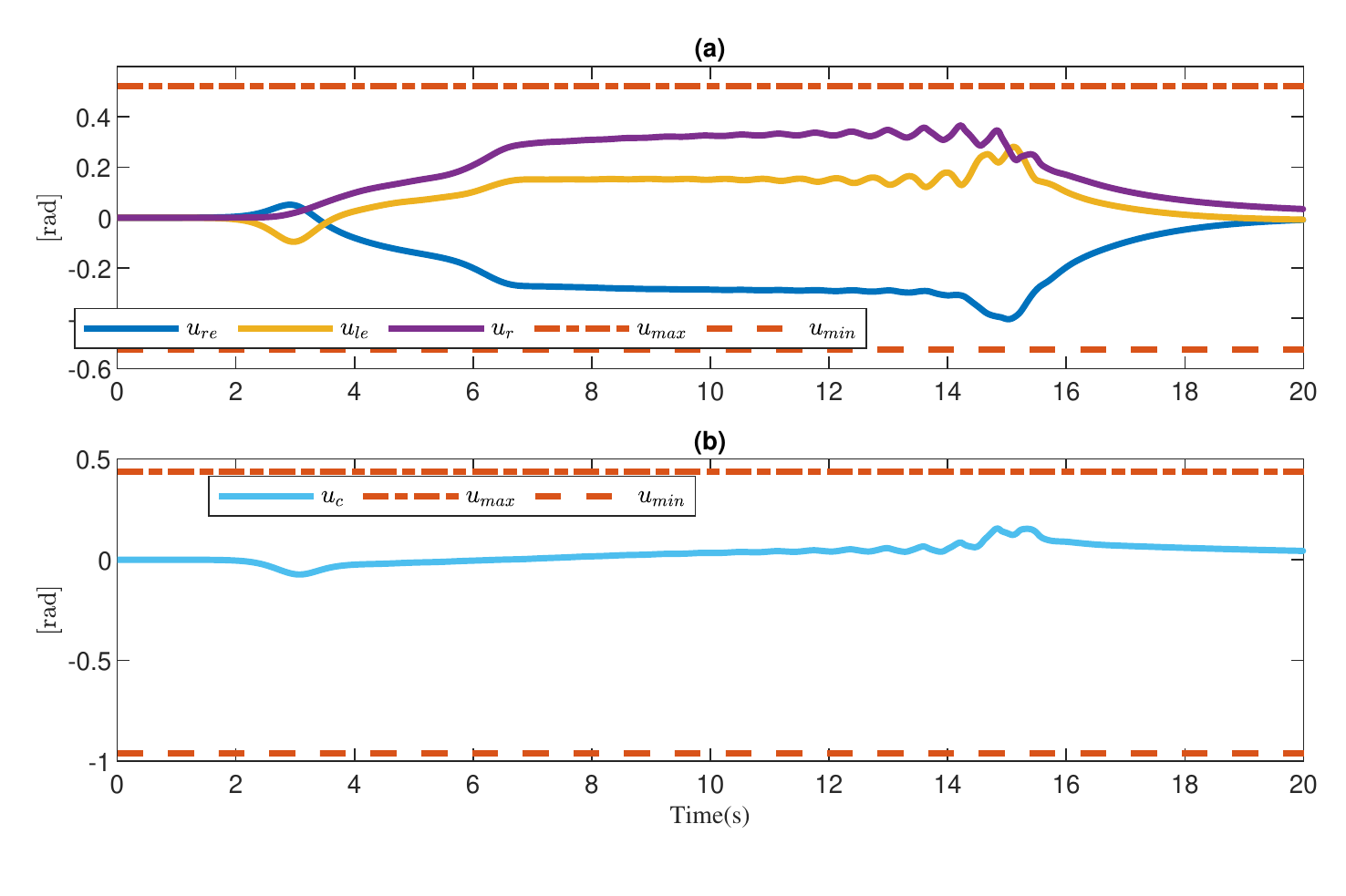}    %
		\caption{Case I: Evolution of the actuator inputs in the presence of \textbf{magnitude saturation}, using the \textbf{conventional projection method}.} 
		\label{fig:u_proj1}
	\end{center}
\end{figure}
\begin{figure}
	\begin{center}
		\includegraphics[width=14cm]{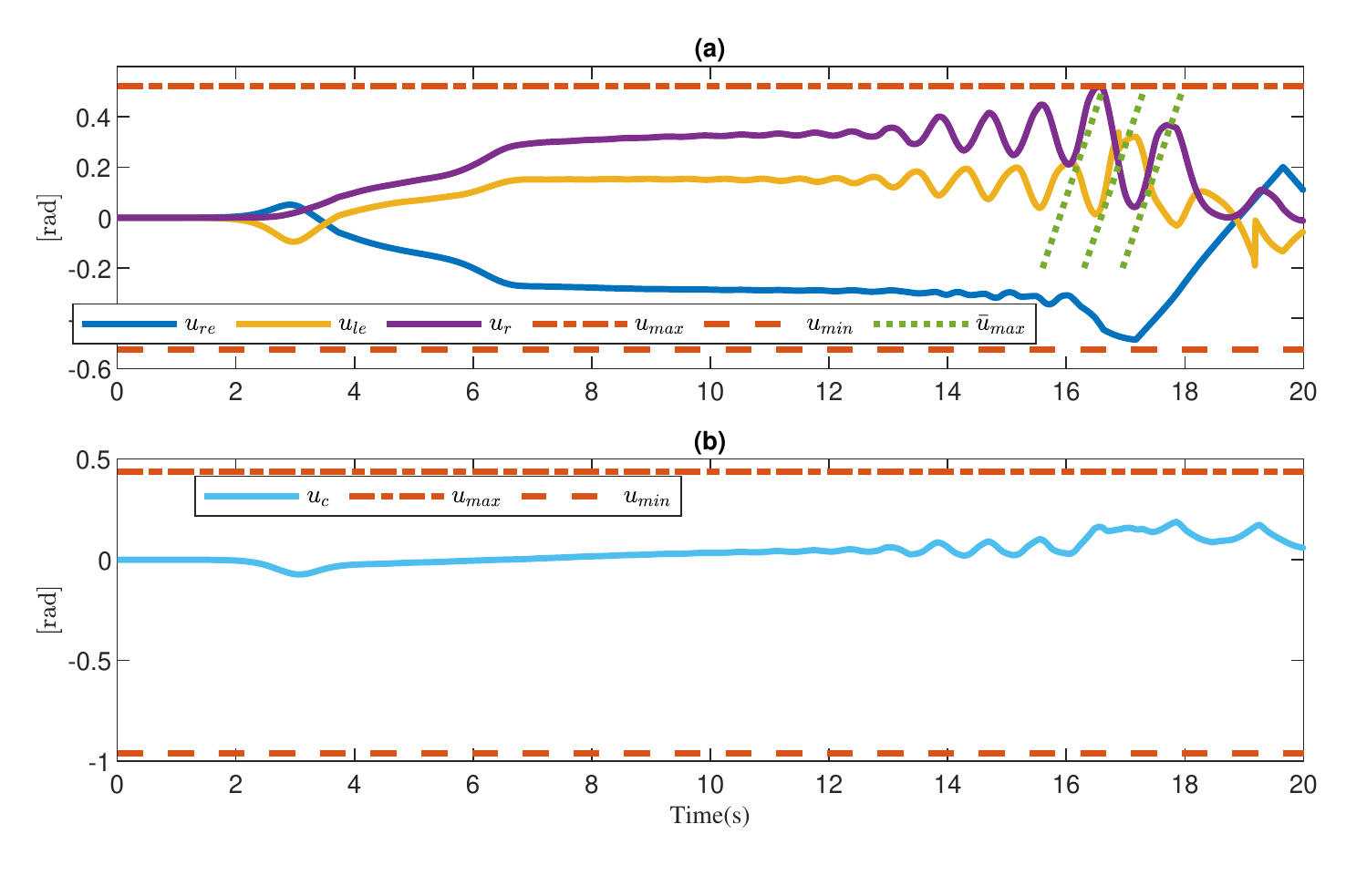}    %
		\caption{Case II: Evolution of the actuator inputs in the presence of \textbf{both magnitude and rate saturation}, using the \textbf{conventional projection method}.} 
		\label{fig:u_proj2}
	\end{center}
\end{figure}
\begin{figure}
	\begin{center}
		\includegraphics[width=14cm]{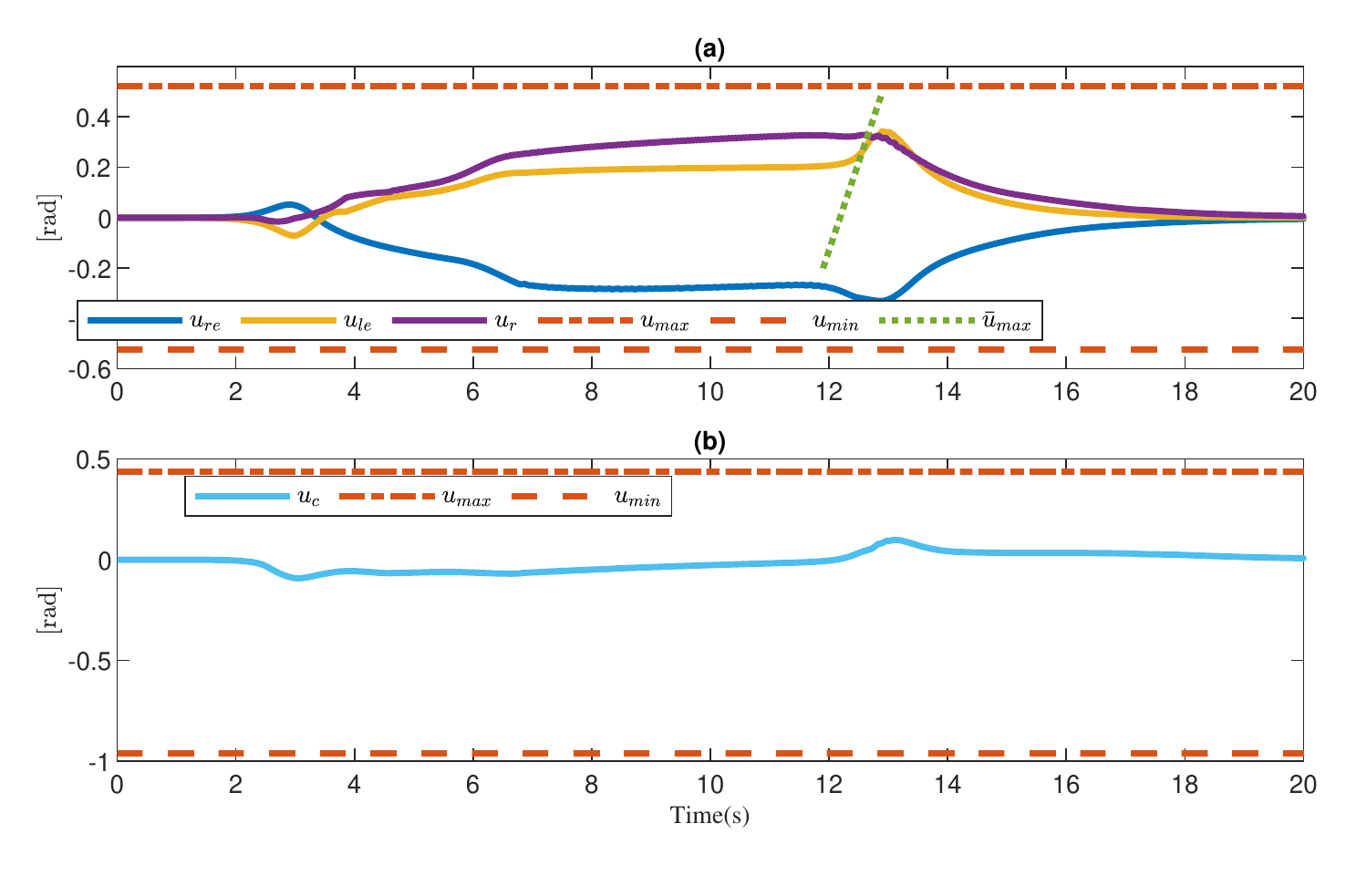}    %
		\caption{Case III: Evolution of the actuator inputs in the presence of \textbf{both magnitude and rate saturation}, using the \textbf{proposed projection algorithm}.} 
		\label{fig:u_proj3}
	\end{center}
\end{figure}

\subsection{Simulation results}
The closed loop control structure depicted in Figure \ref{fig:figure3} is used for the simulations. The reference signal is $ ref=[p_{ref},\ q_{ref},\ r_{ref}]^T $, where $ p_{ref},\ q_{ref} $ and $ r_{ref} $ are the desired roll, pitch and yaw rates, respectively. The effectiveness of the actuators are reduced by $ 30\% $ at $ t=6 $s. 


Three different cases are simulated. Figure \ref{fig:sys11} shows the evolution of the system states, total control input signals, $ v_i,\ i=1, 2, 3 $, and the adaptive parameters, $ \theta_v $,  in the presence of actuator magnitude saturation and conventional projection algorithm (\ref{eq:17n}). It is seen that all the signals are bounded and $ p,\ q $ and $ r $ track their references. Also, the total control input $ v $ is realized reasonably well.

In the second case, actuators are both magnitude and rate limited and again the conventional projection algorithm is used. It is shown in Figure \ref{fig:sys22} that the overall closed loop system shows oscillatory behavior under these conditions. 
 
Finally, in the third case, the proposed projection algorithm is applied in the presence of both magnitude and rate saturation. Figure \ref{fig:sys33} demonstrates the resulting stable and oscillation-free system response. 

The effect of the conventional and the proposed projection algorithms on the actuator input signals are presented separately, in Figures \ref{fig:u_proj1}-\ref{fig:u_proj3}, to emphasize the ability of the latter to limit the signal rates. Figure \ref{fig:u_proj1} shows that the conventional projection algorithm is able to limit the actuator signals within predefined values, when the actuators are only magnitude limited. When actuators are both magnitude and rate limited, the conventional projection algorithm fails to limit the rate of change of actuator signals. This is shown in Figure \ref{fig:u_proj2}, where $ u_{le} $ (yellow line) and $ u_r $ (purple line) increase faster than the rate limit (dashed green line). Finally, Figure \ref{fig:u_proj3} shows that the proposed projection algorithm is capable of limiting both the magnitude and the rate of actuator signals. This can be deduced from the observation that the rate of change of the fastest growing actuator signal, $ u_{le} $ (yellow line), grows still slower than the rate limit (dashed green line). 




\section{Summary}\label{sec7}

A modified projection algorithm that is capable of bounding both the magnitude and rate of change of adaptive parameters is proposed in this paper. This method can be combined with an adaptive control allocator for the control of uncertain over-actuated systems with constrained actuators. The existence and uniqueness of the solutions of the differential equation describing the proposed projection algorithm are shown. Furthermore, properties of the modified projection algorithm that are instrumental for the stability analysis are proven. The performance of the exploited control allocator, in terms of the error bounds, is also guaranteed with the help of the presented projection method. The simulation results with the ADMIRE aircraft model are provided to demonstrate the efficacy of the proposed algorithm.

%
%
%
\section*{Disclosure statement}
No potential conflict of interest was reported by the authors.
%
%
%
\section*{Funding}
This work was supported by the Scientific and Technological Research Council of Turkey under grant number 118E202, and by the Turkish Academy of Sciences Young Scientist Award Program.
%
%
%
%
%
%
%
%
%
%

\bibliographystyle{apacite}
\bibliography{References}

\end{document}